\documentclass[conference,letterpaper]{IEEEtran}

\usepackage[utf8]{inputenc} 
\usepackage[T1]{fontenc}
\usepackage{url}
\usepackage{ifthen}
\usepackage{cite}
\usepackage[cmex10]{amsmath} % Use the [cmex10] option to ensure complicance
\usepackage{amssymb,amsfonts}
\usepackage{algorithmic}
\usepackage{graphicx}
\usepackage{textcomp}
\usepackage{xcolor}
\usepackage{lipsum}
\usepackage{changepage}
\usepackage{amsthm}
\usepackage{acronym}

\usepackage{siunitx} %tabular S align by . (table-format=+x.y)
\usepackage{bm}
\usepackage{float}
\usepackage{multirow}
\usepackage{bbm}
\usepackage{cancel}
\usepackage{theoremref}
\usepackage{hyperref}
\usepackage{BOONDOX-cal}
\usepackage[numbers,square]{natbib}
\usepackage{enumitem}

\theoremstyle{plain}
\newtheorem{thm}{Theorem}[section]

\newtheorem*{cor}{Corollary}

\newtheorem{lem}[thm]{Lemma}

\theoremstyle{definition}
\newtheorem{defn}{Definition}[section]

\theoremstyle{remark}

\newcommand{\eb}{\begin{enumerate}}
	\newcommand{\ee}{\end{enumerate}}

\newcommand{\ceil}[1]{\left\lceil #1 \right\rceil}
\newcommand{\floor}[1]{\left\lfloor #1 \right\rfloor}

\newcommand{\E}[1]{\mathbb{E}\left[#1\right]}
\newcommand{\EE}[2]{\mathbb{E}_{#1}\left[#2\right]}
\newcommand{\abs}[1]{\left| #1 \right|}
\newcommand{\Perp}{\perp\kern-4.5pt\perp}
\newcommand{\seq}{\ensuremath{\mathbf}}

\allowdisplaybreaks

\begin{document}
\title{Low-Rate, Low-Distortion Compression \\ with Wasserstein Distortion}

% %%% Single author, or several authors with same affiliation:
 \author{%
  \IEEEauthorblockN{Yang Qiu and Aaron B.~Wagner}
  \IEEEauthorblockA{School of Electrical and Computer Engineering\\
                     Cornell University\\
                     Ithaca, NY 14853 USA\\
                     Email: \{yq268,wagner\}@cornell.edu}
 }

%%% Several authors with up to three affiliations:
% \author{%
%   \IEEEauthorblockN{Yang Qiu and Aaron B.~Wagner}
%   \IEEEauthorblockA{School of Electrical and Computer Engineering\\
%                     Cornell University\\
%                     Ithaca, NY, USA\\
%                     Email: \{yq268,wagner\}@cornell.edu}
% }
%  \and
%  \IEEEauthorblockN{Claude E.~Shannon and David Slepian}
%  \IEEEauthorblockA{Bell Telephone Laboratories, Inc.\\ 
%                    Murray Hill, NJ, USA\\
%                    Email: \{csh, dsl\}@bell-labs.com}

%%% Many authors with many affiliations:
% \author{%
%   \IEEEauthorblockN{Andrew R.~Barron\IEEEauthorrefmark{1},
%                     Claude E.~Shannon\IEEEauthorrefmark{2},
%                     David Slepian\IEEEauthorrefmark{2},
%                     and Jacob Ziv\IEEEauthorrefmark{2}\IEEEauthorrefmark{3}}
%   \IEEEauthorblockA{\IEEEauthorrefmark{1}%
%                    Department of Statistics and Data Science, Yale University, New Haven, CT, USA,
%                     andrew.barron@yale.edu}
%   \IEEEauthorblockA{\IEEEauthorrefmark{2}%
%                     Bell Telephone Laboratories, Inc.,
%                     Murray Hill, NJ, USA,
%                     \{csh,dsl,jz\}@bell-labs.com}
%   \IEEEauthorblockA{\IEEEauthorrefmark{3}%
%                     Department of Electrical Engineering, Technion---Institute of Technology, Haifa, Israel,
%                     jz@ee.technion.ac.il}
% }

\maketitle

%%%%%%
%% Abstract: 
%% If your paper is eligible for the student paper award, please add
%% the comment "THIS PAPER IS ELIGIBLE FOR THE STUDENT PAPER
%% AWARD." as a first line in the abstract. 
%% For the final version of the accepted paper, please do not forget
%% to remove this comment!
%%

\begin{abstract}
%	Rate-distortion-realism tradeoff has been popular over the past few years. 
%    Classical rate-distortion theory aims to produce reconstructions that are close to source pixel-wise; recently, realism conditions, also known as perceptual quality, have also attracted attention, where the reconstructions seek to look `realistic' to human observers. 
	Wasserstein distortion is a one-parameter family of distortion 
    measures that was recently proposed to unify fidelity 
    and realism constraints. %bridge the gap between the two. 
    After establishing continuity results for Wasserstein in the extreme
    cases of pure fidelity and pure realism, we prove the first coding
    theorems for compression under Wasserstein distortion focusing on
    the regime in which both the rate and the distortion are small.
	%In this work, we study the theoretical properties for Wasserstein distortion. We first show that Wasserstein distortion reduces to fidelity constraint and realism constraint as special cases. We then study the low-rate scheme for Wasserstein distortion. In particular, we  provide two achievability schemes and a converse bound, in the event of rate going to zero.
\end{abstract}

\section{Introduction}

In classical rate-distortion theory, one seeks to represent each source
sequence with as few bits as possible, while producing reconstructions that are close to the source under some distortion metric. Standard metrics, such as PSNR, SSIM~\cite{wang2004image}, etc.~\cite{avcibas2002statistical,dosselmann2005existing,hore2010image}, result in reconstructions that preserve high sample-level fidelity to the source sequences. While these metrics have proven useful~\cite{berger1971rate,pearlman2011digital,sayood2017introduction} in image compression, the reconstructions produced under these metrics are prone to artifacts~\cite{wang2009mean}. Similar deficiencies are also seen in tasks such as image deblurring~\cite{nah2021ntire}, denoising~\cite{buades2005review}, 
and super-resolution~\cite{kwon2015efficient}.

Recently, a new type of constraint, namely \emph{realism}\footnote{Also known as \emph{perceptual quality} in some literature.}, has been proposed to combat such defects. Realism refers to a distribution-level distance between the source and reconstruction~\cite{blau2018perception} (see also~\cite{delp1991moment,li2011distribution,saldi2014randomized}). For instance, the distribution for the source (resp., reconstruction) could be the ensemble distribution of the source images (resp., reconstructed images); then minimizing the distributional distance would reduce blurriness and other artifacts, as a distribution over the space of crisp source images and another over the space of blurry images would have a sizable distributional distance. 

Realism constraints have been extensively studied in recent years, both experimentally~\cite{rippel2017real,tschannen2018deep,agustsson2019generative,mentzer2020high} and theoretically~\cite{klejsa2013multiple,blau2018perception,blau2019rethinking,matsumoto2018introducing,matsumoto2019rate,theis2021coding,chen2021distribution,chen2022rate,wagner2022rate,hamdi2023rate}. 
Past work has considered distributions induced from the source and 
reconstructions 
in various ways, such as by considering the distribution of a full-sized,
randomly chosen image~\cite{theis2021coding,theis2022lossy, wagner2022rate, chen2022rate, hamdi2023rate}; the distribution of a random patch from a randomly
selected image~\cite{agustsson2019generative}; or the distribution of
a random patch from a single image~\cite{wang2018esrgan,gao2021perceptual}.

Existing studies treat fidelity and realism as distinct constraints; some have even argued that they are in tension~\cite{blau2018perception,zhang2021universal,chen2022rate,niu2023conditional,salehkalaibar2023rate}, resulting in the \emph{distortion-perception tradeoff}. Yet the ultimate goal for both notions are identical: to quantify the differences between two images perceived by human observers. Thus a unified notion of distortion that simultaneously generalizes fidelity
and realism is desirable---such a generalization has the potential to
better capture human-perceived distortion between images than either
can alone. Recently, a new one-parameter family of distortion measures, \emph{Wasserstein distortion}~\cite{qiu2023wasserstein}, was proposed as a simultaneous
generalization of the two. 

Wasserstein distortion is inspired by models of the Human Visual System (HVS), namely the \emph{summary statistics} model proposed in~\cite{balas2009summary,rosenholtz2011your,rosenholtz2012summary}. The model is described in  detail in~\cite{freeman2011metamers}. In summary, the model describes how information is processed in the first two areas of the ventral stream. 
Specifically, it is assumed that the vision system computes statistics
of filter responses over different \emph{receptive fields}.
%The V1 neurons pass the current field of view through multiple oriented filters with different orientations and spatial frequencies; the V2 neurons then compute higher order statistics from the V1 outputs over various receptive fields. 
The receptive fields grow with distance from the fovea, or center of the
gaze, as depicted in Fig.~\ref{fig:v1}. In the visual periphery, the receptive fields are large and, as a result, statistical information is pooled
over a large area. In the fovea, the receptive field is small enough that the original image can be recovered from the statistics. 

Wasserstein distortion translates this model into a distortion measure.
A given image is covered with overlapping \emph{pooling regions} of potentially
 different sizes. Within each region, a distribution over filter responses is defined. The Wasserstein distortion between two images is then defined to be the spatially-averaged distance between the corresponding distributions between the two images. It is controlled by a width
parameter ($\sigma$ to follow) that controls the size of the pooling region:
when $\sigma$ is large, the statistics are pooled over large regions,
akin to a realism measure.
When $\sigma$ is small, the statistics are pooled over a small
region, so that Wasserstein distortion reduces to a conventional
pixel-level fidelity measure. Note that the $\sigma$ parameter
can vary spatially over the image, with some portions subject to a
fidelity constraint, others subject to a realism constraint, and others
subject to a constraint that is a fusion of the two.
\begin{figure}[ht]
    \centering
    \includegraphics[width=0.4\linewidth]{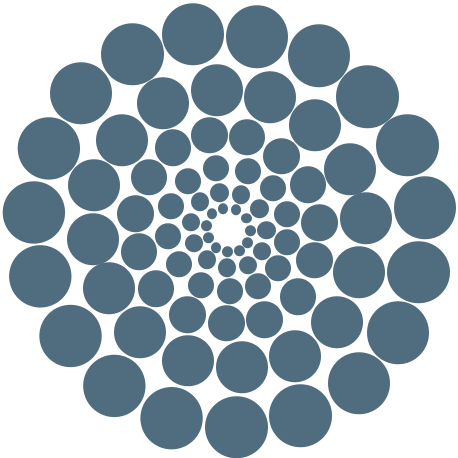}
    \caption[Receptive fields in the ventral stream grow with eccentricity.]{\protect\raggedright Receptive fields in the ventral stream grow with eccentricity.}
    \label{fig:v1}
\end{figure}

Wasserstein distortion was introduced and experimentally validated in~\cite{qiu2023wasserstein}. This work considers the metric from a theoretical
viewpoint. We first show that Wasserstein distortion reduces to 
fidelity and realism constraints in a continuous way as 
$\sigma \rightarrow 0$ and $\sigma \rightarrow \infty$, respectively.
We then consider coding theorems for i.i.d.\ sources under
Wasserstein distortion focusing on the large-$\sigma$ regime. 
For stationary ergodic sources, in the limit as $\sigma$ tends 
to infinity, zero distortion can be achieved with zero rate,
since the decoder can simply output an independent realization
of the source. We show that, under certain design choices for
Wasserstein distortion, as $\sigma \rightarrow \infty$,
the rate and distortion vanish as $(\frac{1}{\sigma^\alpha},
\frac{1}{\sigma^\beta})$ and partially characterize the optimal
tradeoff between $\alpha$ and $\beta$.

The balance of the paper is organized as follows. Section~\ref{sec:setup} consists of a self-contained description of Wasserstein distortion. Section~\ref{sec:extreme} proves that Wasserstein distortion continuously reduces to fidelity and realism constraints as $\sigma$ tend to zero and infinity, respectively. Section~\ref{sec:lowrate} provides the rate-distortion analysis for Wasserstein distortion in the large-$\sigma$ regime, where two achievable schemes and one converse argument are provided.

\section{Wasserstein distortion}
\label{sec:setup}

Let $\seq{X} = \{X_n\}_{n = -\infty}^\infty$ be a stochastic process that represents the source of interest, with realizations denoted by $\seq{x} = \{x_n\}_{n = -\infty}^\infty$. 

Let $T$ denote the unit advance operation, i.e., if $\seq{x}' = T \seq{x}$ then
\begin{equation}
    x_{n}' = x_{n+1}.
\end{equation}
We denote the $k$-fold composition $T \circ T \circ \ldots \circ T$ by $T^k$. We assume that $\seq{X}$ is strongly stationary, i.e., $T\seq{X} \overset{d}{=} \seq{X}$.%, and ergodic, i.e., for a set $A$ of sequences, if $T^{-1}(A) = T$ then $Pr(\seq{X} \in A) = 0$ or $1$. 

Let $\phi(\seq{x})$ denote a vector of local features of $\{x_n\}_{n = -\infty}^\infty$ about $n = 0$. Define the random variable $Z$ by $Z = \phi(\seq{X})$ and the process $\seq{Z}$ by
\begin{equation}
    Z_n = \phi(T^{n} \seq{X}).
\end{equation}
Then $\seq{Z}$ is also strongly stationary. The $\phi$ function models the lower-level ventral streams in the retina that capture features in the field of vision, as depicted above. $\phi$ can take many forms: the coordinate map,  convolution with multiple kernels, steerable pyramid~\cite{portilla2000parametric,freeman2011metamers}, convolution with random kernels followed by non-linearity~\cite{ustyuzhaninov2017does}, selected layers of a convolutional neural network~\cite{gatys2015texture}, etc.

Let $q_\sigma(k)$, $k \in \mathbb{Z}$, denote a family of probability mass functions
(PMFs)
over the integers, parameterized by $0 \leq \sigma < \infty$, satisfying~\cite{qiu2023wasserstein}:
\begin{enumerate}[label=\textbf{P.\arabic*}]
	\item \label{TSG:symm} For any $\sigma$ and $k$, $q_\sigma(k) = q_\sigma(-k)$; 
	\item \label{TSG:mono} For any $\sigma$ and $k,k' \in \mathbb{Z}$ such that $\abs{k} \leq \abs{k'}$, $q_\sigma(k) \geq q_\sigma(k')$;
	\item \label{TSG:delta} If $\sigma = 0$, $q_\sigma$ is the Kronecker delta function, i.e., 
	   $q_0(k) = \begin{cases} 1 & k = 0 \\ 0 & k \neq 0\end{cases}$;
	\item \label{TSG:cont} For all $k$, $q_\sigma(k)$ is continuous in $\sigma$
	    at $\sigma = 0$;
	\item \label{TSG:tail} There exists $\epsilon > 0$ and $K$ so that for all $k$ such that $|k| \ge K$,
	$q_\sigma(k)$ is nondecreasing in $\sigma$ over the range $[0,\epsilon]$; and
	\item \label{TSG:unif} For any $k$, $\lim_{\sigma \rightarrow \infty} q_{\sigma}(k) = 0$.
\end{enumerate}
We call $q_\sigma(\cdot)$ the \emph{pooling PMF} and $\sigma$ the \emph{pooling width} or \emph{pooling parameter}. Our setup for Wasserstein distortion is agnostic to the choice of the pooling PMF, as long as~\ref{TSG:symm} --~\ref{TSG:unif} are satisfied. In Section~\ref{sec:extreme}, we do not presume any specific PMF. In Section~\ref{sec:lowrate}, the particular PMF that we consider is the two-sided geometric distribution,
\begin{equation}
q_\sigma(k) = \begin{cases}
\frac{e^{1/\sigma} - 1}{e^{1/\sigma} + 1} \cdot e^{-|k|/\sigma} & \text{if $\sigma > 0$} \\
1 & \text{if $\sigma = 0$ and $k = 0$} \\
0 & \text{otherwise}.
\end{cases}
\end{equation} 
One can verify that $q_\sigma(k)$ satisfies~\ref{TSG:symm} --~\ref{TSG:unif}.

Given a realization $\seq{x}$, we define a sequence of measures $\seq{y} = \{y_n\}_{n = -\infty}^\infty$ via
\begin{equation}
    y_n = \sum_{k = -\infty}^\infty q_{\sigma}(k) \delta_{z_{n+k}},
	\label{eq:measure}
\end{equation}
where $\delta_\cdot$ denotes the Dirac delta function. Then $\seq{Y}$ is a measure-valued random process, i.e., for each $n$, $Y_n$ is a random measure. Each realization $y_n$ represents the statistics of the features pooled across a particular receptive field. The effective width of the receptive field are proportional to the parameter $\sigma$. See Fig.~\ref{fig:illustrate} for a pictorial illustration for the process.

\begin{figure*}[ht]
	\centering
	\includegraphics[width=0.9\linewidth]{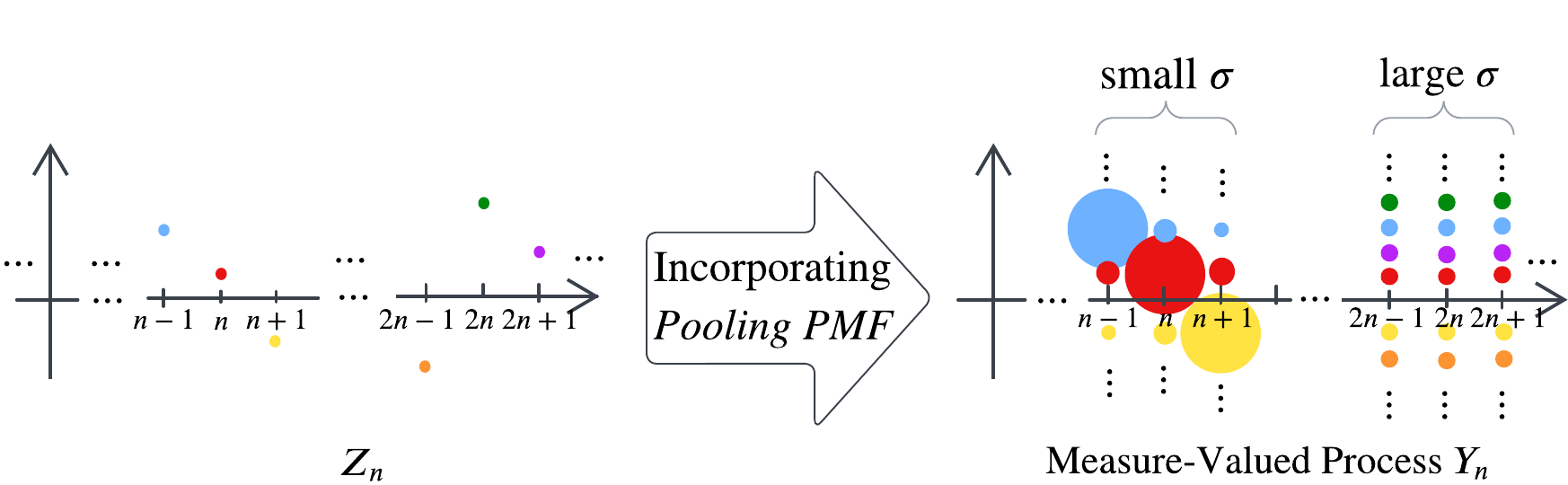}
	\label{fig:illustrate}
	\caption{A pictorial illustration of (\ref{eq:measure})~\cite{qiu2023wasserstein}. In the right plot, the size of the disk indicates the probability mass and the vertical coordinate of the center of the disk indicates the value.}
\end{figure*}

Similarly, we can define $\hat{\seq{x}} = \{\hat{x}_n\}_{n = -\infty}^\infty$, $\hat{\seq{z}} = \{\hat{z}_n\}_{n = -\infty}^\infty$, $\hat{\seq{y}} = \{\hat{y}_n\}_{n = -\infty}^\infty$, etc., for the reconstruction process. 

Consider any divergence between distributions $\mathcal{D}(\rho,\rho')$ over Euclidean space of a given dimension. Then our distortion measure at time $n$ is defined to be
\begin{equation}
	D_n = \mathcal{D}\left(y_n,\hat{y}_n\right).
\end{equation}

The Wasserstein distortion $D$ over a  block $\{-N,\ldots,N\}$ is defined as a spatial average
\begin{equation}
    D = \frac{1}{2N+1} \sum_{n = -N}^N D_n.
\end{equation}

Choices of the divergence $\mathcal{D}$ include $p$-Wasserstein distance~\cite{villani2009optimal} to the $p$-th power, sliced Wasserstein distance~\cite{pitie2005sliced,bonneel2015sliced,tartavel2016wasserstein,heitz2021wasserstein}, Sinkhorn distance~\cite{cuturi2013sinkhorn}, Maximum Mean Discrepancy (MMD)~\cite{smola2006maximum,li2017mmd,li2019implicit}, or the distance between Gram matrices~\cite{gatys2015texture,ustyuzhaninov2017does}. In this work, we choose $\mathcal{D}$ to be the $p$-Wasserstein distance~\cite[Def.~6.1]{villani2009optimal}\footnote{We refer to $W_p$ as the Wasserstein \emph{distance} even though it is not necessarily a metric if $d$ is not a metric.} to the $p$-th power. Let $d : \mathcal{Z} \times \hat{\mathcal{Z}} \mapsto [0,\infty)$ be a cost function. The $p$-Wasserstein distance induced by $d$ for distributions $\rho$ and $\rho'$ is 
\begin{equation}
    W_p(\rho,\rho') =\left( \min_{X \sim \rho, X' \sim \rho'} \E{d^p(X,\hat{X})}\right)^{1/p}.
\end{equation}
In Section~\ref{sec:lowrate}, we take $p = 2$.

\section{Fidelity and Realism as Extreme Cases}
\label{sec:extreme}

Let $\seq{x}$ and $\hat{\seq{x}}$ be two sequences and let $\seq{z}$ and $\hat{\seq{z}}$ denote the associated feature sequences, i.e., $z_n = \phi(T^n \seq{x})$ and $\hat{z}_n = \phi(T^n \hat{\seq{x}})$. If one is only concerned with fidelity to the original image, one might use an objective such as
\begin{equation}
\label{eq:reduce:fidelity}
\frac{1}{2N+1} \sum_{n = -N}^N d^p(z_n,\hat{z}_n),
\end{equation}
perhaps with $\phi$ being the identity map; conventional mean squared error can be expressed in this way with $p = 2$. This objective can be trivially recovered from Wasserstein distortion by taking $\sigma = 0$, invoking \ref{TSG:delta}, and applying the formula for the Wasserstein distance between point masses:
\begin{equation}
W_p(\delta_z,\delta_{\hat{z}}) = d(z,\hat{z}).
\end{equation}
Given that we are interested in smoothly interpolating between fidelity and realism, we would like Wasserstein distortion to reduce to (\ref{eq:reduce:fidelity}) in the limit as $\sigma \rightarrow 0$. We next identify conditions under which this continuity result holds. Note that this result does not require $d$ to be a metric.
\begin{thm}
Suppose $q$ satisfies \ref{TSG:delta} -- \ref{TSG:tail} and $\seq{z}$, $\hat{\seq{z}}$, and $q$ together satisfy
\begin{equation}
\label{eq:fidelity:hyp}
   \sum_{k = -\infty}^\infty q_\sigma(k) d^p(z_k,\hat{z}_k) < \infty
\end{equation}
for all $\sigma > 0$. Then we have
\begin{equation}
\lim_{\sigma \rightarrow 0} D_{0,\sigma} = d^p(z_0,\hat{z}_0).
\end{equation}
\label{thm:fidelity}
\end{thm}

Likewise, we show that Wasserstein distortion continuously reduces to pure realism in the large-$\sigma$ limit. We use $\stackrel{w}{\rightarrow}$ to denote weak convergence.

\begin{thm}
Assume that $\seq{X}$ is ergodic, i.e., for a set $A$ of sequences, if $T^{-1}(A) = T$ then $Pr(\seq{X} \in A) = 0$ or $1$. Suppose $q$ satisfies \ref{TSG:symm}, \ref{TSG:mono}, and \ref{TSG:unif} and $d$ is a metric.  Let $F_N$ (resp.\ $\hat{F}_N$) denote the empirical CDF of $\{z_{-N},\ldots,z_N\}$ (resp.\ $\{\hat{z}_{-N},\ldots,\hat{z}_N\}$) and suppose we have
\begin{align}
\label{eq:realism:CDF}
	F_N \stackrel{w}{\rightarrow} F \ &\text{and} \ \hat{F}_N \stackrel{w}{\rightarrow} \hat{F}; \\
\label{eq:realism:moment}
	\int d^p(z,0) dF_N &\rightarrow 
	\int d^p(z,0) dF < \infty \nonumber \\
\text{and }
	\int d^p(z,0) d\hat{F}_N &\rightarrow 
	\int d^p(z,0) d\hat{F} < \infty; \\
\intertext{and, for all $\sigma$,} 
   \sum_{k = -\infty}^\infty q_\sigma(k) &d^p(z_k,0) < \infty  \nonumber \\
\text{and }  
	\sum_{k = -\infty}^\infty q_\sigma(k) &d^p(\hat{z}_k,0) < \infty.
	\label{eq:realism:summable}
\end{align}
%\begin{equation}
%\label{eq:averagemoment}
%\lim_{N \rightarrow \infty} 
%\frac{1}{2N + 1} \sum_{n = -N}^N d^p(z_n,0) = \int d^p(z,0) dF(z)
%\end{equation}
%and likewise for $\hat{\seq{z}}$ and $\hat{F}$. 
Then we have
\begin{equation}
\lim_{\sigma \rightarrow \infty} D_{0,\sigma} = W_p^p(F,\hat{F}).
\end{equation}
\label{thm:realism}
\end{thm}

It follows from the previous result that when the source ensemble is ergodic, as occurs with textures, then in the large-$\sigma$ limit Wasserstein distortion reduces to the ensemble form of realism. That is, it equals the $p$-Wasserstein distance to the $p$-th power between the true distributions of the images and reconstructions, denoted by $F$ and $\hat{F}$ in the following corollary.

\begin{cor}
Suppose $q$ satisfies \ref{TSG:symm}, \ref{TSG:mono}, and \ref{TSG:unif} and $d$ is a metric. Suppose $\seq{X}$ and $\hat{\seq{X}}$ are stationary ergodic processes and let 
$F$ (resp.\ $\hat{F}$) denote the CDF of $Z_0$ (resp.\ $\hat{Z}_0$). If
\begin{equation} \E{d^p(Z_0,0)} < \infty \ \text{and} \ 
\E{d^p(\hat{Z}_0,0)} < \infty,
\label{eq:finitepthmoment}
\end{equation}
then we have
\begin{equation}
\lim_{\sigma \rightarrow \infty} D_{0,\sigma} = W_p^p(F,\hat{F}) \quad \text{a.s.}
\end{equation}
\label{cor:ergodic}
\end{cor}

For all proofs in this section, please see Appendix~\ref{app:proofs}.

\section{Low-rate Scheme Analysis}
\label{sec:lowrate}	

We turn to the problem of optimal compression under Wasserstein 
distortion. We assume that $\seq{X}$ is a doubly-infinite process, and $X_n$'s are i.i.d.\ over a finite alphabet. The code will be defined over a block of length $2N+1$, as described below. To obtain $\seq{\hat{X}}$, divide the time horizon into blocks of size $2N+1$ and apply the same code separately on each block. The distortion is also calculated within each block. For the remainder of the section, we focus on the `center block', i.e., the block containing index $0$. The same analysis applies to all other blocks.

\begin{defn}
	An $(N,R,\Delta)$-code is an encoder
	\begin{equation}
		f_N: \mathcal{X}_{-N}^N \to \{0,1\}^{\ceil{(2N+1)R}},
	\end{equation}
	and a decoder
	\begin{equation}
		g_N: \{0,1\}^{\ceil{(2N+1)R}} \to \hat{\mathcal{X}}_{-N}^N,
	\end{equation}
	such that, when the code is applied to each of the blocks of length $2N+1$, over the same block,
	\begin{equation}
		\E{D} = \E{\frac{1}{2N+1} W_2^2(Y_n,\hat{Y}_n)} \leq \Delta,
	\end{equation}
	where $X_i \in \mathcal{X}$, $\hat{X}_i \in \hat{\mathcal{X}}$ for all  $i = -N,\ldots,N$, and $(\hat{X}_{-N},\ldots,\hat{X}_N) = g_N(f_N(X_{-N},\ldots,X_N))$. We call $R$ the rate of the code.
\end{defn}

We define the rate-distortion region in the usual sense:

\begin{defn}
	A rate-distortion pair $(R,\Delta)$ is achievable if there exists a sequence of $N_1,N_2,\ldots$, $N_i \to \infty$ as $i \to \infty$, such that when the same $(N_i,R,\Delta)$ code is applied to each block of length $2N_i+1$, $\lim_{i \to \infty} \E{D(X_{-N_i}^{N_i}, g_{N_i}(f_{N_i}(X_{-N_i}^{N_i})))} \leq \Delta$. The rate-distortion region $RD_\sigma$ is the closure of the set of achievable rate-distortion pairs $(R,\Delta)$.
\end{defn}

From Theorem~\ref{thm:realism}, we see that if $\hat{X}_n$ has the same distribution as $X_n$ with the identity mapping being the only kernel $\phi$, our distortion will converge to $0$ when $\sigma \to \infty$. Since the distortion diminishes, we would not need to send any information about the source, hence the optimal rate should also converge to $0$ when $\sigma \to \infty$. In other words, in the large-$\sigma$ scheme, we can achieve low rate and low distortion simultaneously. This is akin to the way that low-rate, low-distortion compression
of textures is possible if one accepts an independent realization of the
texture as a reconstruction. We are interested in studying the tradeoff 
in the speed with which rate and distortion vanish as $\sigma \to \infty$. 
We shall see that the correct scaling for the rate-distortion pair
is $(\sigma^{\alpha}, \sigma^{\beta})$, where $\alpha < 0$ and 
$\beta < 0$, which motivates the following definition.

\begin{defn}
	Let $RD_\sigma$ be the rate-distoriton region with parameter $\sigma$. The pair $(\alpha,\beta)$ is asymptotically achievable if for all sufficiently large $\sigma$, $(\sigma^\alpha,\sigma^\beta) \in RD_\sigma$. Convergence rate region is the closure of the set of achievable pairs $(\alpha,\beta)$.
\end{defn}

We consider the case of discrete alphabets with arbitrary alphabet size, with the only kernel being identity mapping, and we provide the achievability arguments for two schemes and a converse argument. These results partially
characterize the optimal tradeoff between $\alpha$ and $\beta$
and are summarized in Figure~\ref{fig:bound}.

Let $\mathcal{X} = \hat{\mathcal{X}} = \{1,2,\ldots,A\},\ A \geq 2$ be the source/reconstruction alphabet, and let $X_n$ be i.i.d.\ samples drawn according to some distribution $\mathcal{p} = (\mathcal{p}_1,\mathcal{p}_2,\ldots,\mathcal{p}_A)$ over $\{1,2,\ldots,A\}$. Let the underlying distance over $\mathcal{X} \times \hat{\mathcal{X}}$ be defined as
\begin{equation}
	\mathcal{d} = \begin{bmatrix} 0 & d_{1,2} & \cdots & d_{1,A} \\ d_{1,2} & 0 & \cdots & d_{2,A} \\ \vdots & \vdots & \ddots & \vdots \\ d_{1,A} & \cdots & d_{A-1,A} & 0 \end{bmatrix};
\end{equation}
i.e., the distance from a symbol to itself is $0$, and the distance between symbols $i$ and $j$, $i \neq j$ is $d_{i,j} > 0$.

%For simplicity, we always consider the distortion at index $0$. Distortion at other indices can be calculated via a shift of indices, and all results holds. Notice that the squared 2-Wasserstein distance under $\mathcal{d}$ is
%\begin{align}
%W_2^2\left(Y_0,\hat{Y}_0\right) &= \inf_{Y \sim Y_n,\hat{Y} \sim \hat{Y}_n} \E{\mathcal{d}^2\left(Y,\hat{Y}\right)} \nonumber \\
%&= \inf_{\pi \in \Pi\left(Y_0,\hat{Y}_0\right)} \sum_{i,j} d_{i,j}^2 \pi(i,j), \label{eq:wassdefn}
%\end{align}
%where $\Pi(p, q)$ is the collection of all joint distributions with $p$ and $q$ as marginal distributions, as commonly used in Optimal Transport literature.

\begin{figure*}[htp]
    \centering
    \includegraphics[width=0.75\linewidth]{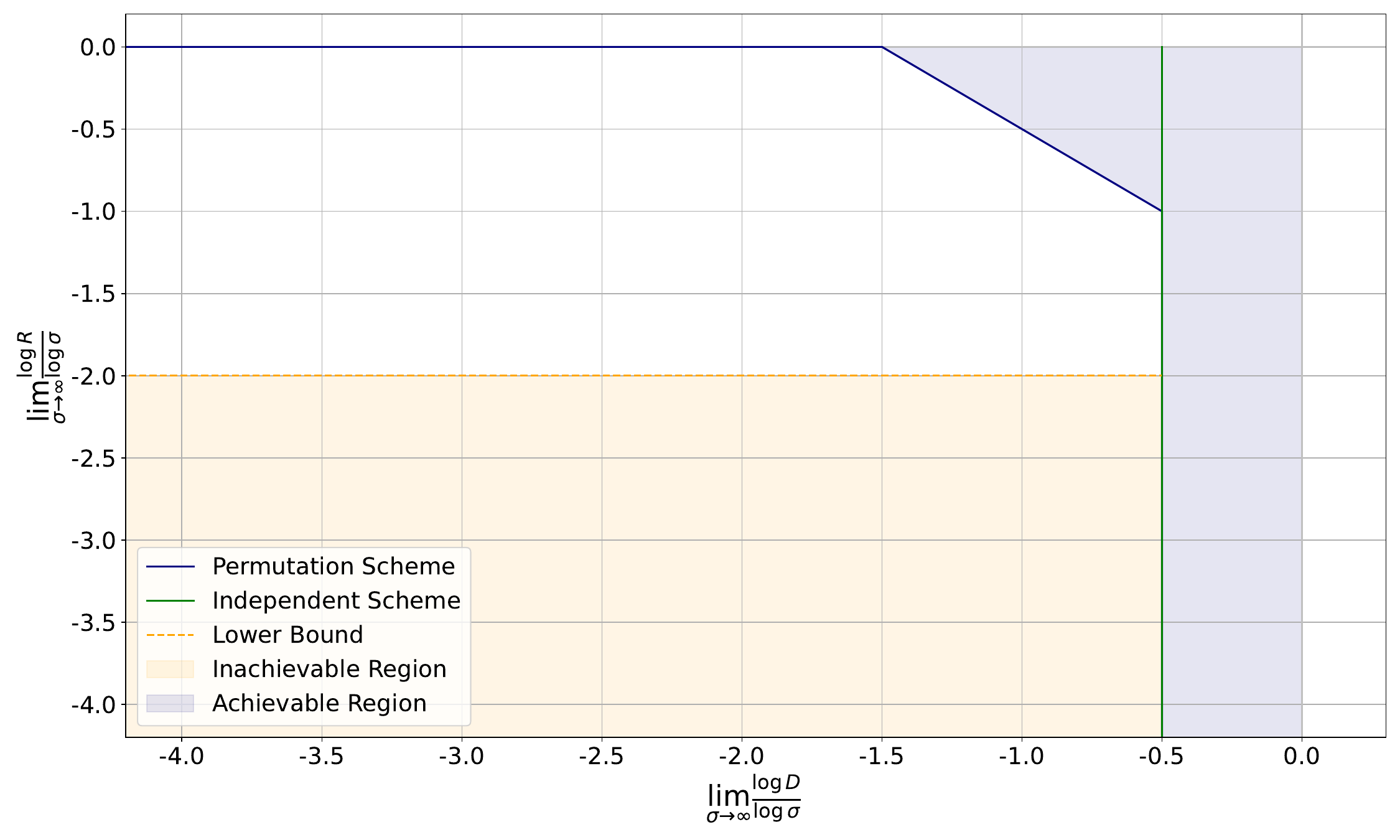}
    \caption{Upper- and Lower- bounds for low-rate scheme.}
    \label{fig:bound}
\end{figure*}

For the achievability and converse proofs, we analyze either $W_{2,\mathcal{d}_{\max}}^2(\cdot,\cdot)$ or $W_{2,\mathcal{d}_{\min}}^2(\cdot,\cdot)$ instead of $W_{2,\mathcal{d}}^2(\cdot,\cdot)$, where $\mathcal{d}_{\max}$ (resp., $\mathcal{d}_{\min}$) is $\mathcal{d}$ with all off-diagonal entries replaced by the maximum off-diagonal entry $d_{\max}$ (resp., minimum off-diagonal entry $d_{\min}$). We write the corresponding Wasserstein distortion as $D_{\mathcal{d}_{\max}}$ (resp., $D_{\mathcal{d}_{\min}}$). The sandwich argument described in Section~\ref{sec:sandwich} would lead the analysis back to $W_{2,\mathcal{d}}^2$ and $D$. For complete details, see Section~\ref{sec:sandwich}.

\subsection{Achievability -- Two Schemes}

In this section, we wish to upper bound $\E{D} = \frac{1}{2N+1} \sum_{n = -N}^N \E{W_2^2(Y_n, \hat{Y}_n)}$. We will analyze $W_{2,\mathcal{d}_{\max}}^2(Y_n,\hat{Y}_n)$ instead of $W_{2,\mathcal{d}}^2(Y_n,\hat{Y}_n)$; the same bounds can be applied on $W_{2,\mathcal{d}_{\min}}^2(Y_n,\hat{Y}_n)$, and we conclude our theorem on $W_{2,\mathcal{d}}^2(Y_n,\hat{Y}_n)$ using the sandwich argument in Section~\ref{sec:sandwich}. In this subsection, $W_2^2(\cdot,\cdot)$ denotes $W_{2,\mathcal{d}_{\max}}^2(\cdot,\cdot)$ unless otherwise specified.

Now, $W_2^2(\cdot,\cdot)$ admits a closed form:~\cite[pp. 10]{villani2009optimal}
\begin{align}
W_2^2(Y_n,\hat{Y}_n) &= \min_{Y \sim Y_n, \hat{Y} \sim \hat{Y}_n} \E{d^2(Y,\hat{Y})} \nonumber \\
&= d_{\max}^2 \cdot \min_{Y \sim Y_n, \hat{Y} \sim \hat{Y}_n} \E{\mathbbm{1}_{Y \neq \hat{Y}}}\\
&= \frac{d^2_{\max}}{2}\cdot \Vert Y_n,\hat{Y}_n\Vert_{\mathrm{TV}} \nonumber \\
&= \frac{d^2_{\max}}{2}\cdot \left(\sum_{i = 1}^A \abs{Y_n(\{i\}) - \hat{Y}_n(\{i\})}\right) \\
&= \frac{d^2_{\max}}{2}\cdot \left(\sum_{i = 1}^A \left|\sum_{m = -\infty}^\infty q_{\sigma}(m-n) \mathbbm{1}_{X_m = i} \right.\right. \nonumber \\
&\phantom{=} \left.\left.- \sum_{m = -\infty}^\infty q_{\sigma}(m-n) \mathbbm{1}_{\hat{X}_m = i}\right|\right).
\end{align}

We consider two schemes, one where the reconstruction $\seq{\hat{X}}$ is an independent realization of $\mathcal{p}$, the other where $\seq{\hat{X}}$ is a random permutation of the source within consecutive windows of a given size.

\subsubsection{Independent Realization}

Assume $\seq{\hat{X}}$ is an independent i.i.d.\ realization of $\mathcal{p}$. 

\begin{thm}
	The pair $(-\infty,\beta)$ is asymptotically achievable for all $\beta > -1/2$.
\end{thm}

\begin{proof}
Choose the sequence of $N$ to be the sequence of integers. The rate of this scheme is $0$ as nothing is transmitted. In this scheme, since both processes $\seq{X}$ and $\seq{\hat{X}}$ are i.i.d., $W_2^2(Y_n,\hat{Y}_n) = W_2^2(Y_0,\hat{Y}_0)$ for all $n = -N,\ldots,N$; we focus on $W_2^2(Y_0,\hat{Y}_0)$:
\begin{align}
	&\E{W_2^2(Y_0,\hat{Y}_0)} = \frac{d^2_{\max}}{2} \times \nonumber \\
	& \left(\sum_{i = 1}^A \mathbb{E} \left[ \left| \sum_{m = -\infty}^\infty q_\sigma(m) \mathbbm{1}_{X_m = i} - \sum_{m = -\infty}^\infty q_\sigma(m) \mathbbm{1}_{\hat{X}_m = i}\right| \right] \right),
\end{align}
%\begin{align}
%	\E{W_2^2(Y_0,\hat{Y}_0)} &= \frac{d^2_{\max}}{2} \cdot \left(\sum_{i = 1}^A \mathbb{E} \left[ \left| \sum_{m} q_\sigma(m) \mathbbm{1}_{X_m = i} \right.\right.\right. \nonumber \\
%    &\phantom{=} \left.\left.\left.- \sum_{m} q_\sigma(m) \mathbbm{1}_{\hat{X}_m = i}\right| \right] \right),
%\end{align}
and for each $i = 1,2,\ldots,A$,
\begin{align}
	&\quad \E{\abs{\sum_{m = -\infty}^\infty q_\sigma(m) \mathbbm{1}_{X_m = i} - \sum_{m = -\infty}^\infty q_\sigma(m) \mathbbm{1}_{\hat{X}_m = i}}} \\
	&\leq \left\{ \mathbb{E} \left[ \left( \sum_{m = -\infty}^\infty q_\sigma(m) (\mathbbm{1}_{X_m = i} - \mathcal{p}_i) \right.\right.\right. \nonumber \\
    &\phantom{\leq} \left.\left.\left.  - \sum_{m = -\infty}^\infty q_\sigma(m) (\mathbbm{1}_{\hat{X}_m = i} - \mathcal{p}_i)\right)^2 \right] \right\}^{1/2} \\
	&= \sqrt{2\E{\left(\sum_{m = -\infty}^\infty q_\sigma(m) (\mathbbm{1}_{X_m = i} - \mathcal{p}_i)\right)^2}} \label{eq:indsamedist}\\
	&= \sqrt{2\sum_{m = -\infty}^\infty q^2_\sigma(m) \E{\left(\mathbbm{1}_{X_m = i} - \mathcal{p}_i\right)^2}} \label{eq:xmiid} \\
	&= \sqrt{2\mathcal{p}_i(1-\mathcal{p}_i) \sum_{m = -\infty}^\infty q^2_\sigma(m)} \\
    &= \sqrt{2\mathcal{p}_i(1-\mathcal{p}_i)\frac{\exp\left(4/\sigma\right)-1}{\left(\exp\left(1/\sigma\right)+1\right)^4}},
\end{align}
where \eqref{eq:indsamedist} holds since $\seq{X}$ and $\seq{\hat{X}}$ are independent with same distribution, and \eqref{eq:xmiid} holds since $\seq{X}$ is i.i.d.

Hence for all $n = -N,\ldots,N$,
\begin{align}
    &\phantom{\leq} \E{W_2^2(Y_n,\hat{Y}_n)} \\ 
    &\leq \sqrt{\frac{\exp\left(4/\sigma\right)-1}{\left(\exp\left(1/\sigma\right)+1\right)^4}}\cdot \frac{d_{\max}^2}{2} \cdot \left(\sum_{i = 1}^A \sqrt{2\mathcal{p}_i(1-\mathcal{p}_i)}\right), \nonumber 
\end{align}
which implies
\begin{align}
    &\phantom{\leq} \E{D_{\mathcal{d}_{\max}}} = \frac{1}{2N+1} \sum_{n = -N}^N \E{W_2^2(Y_n,\hat{Y}_n)} \label{eq:indepschemebound}\\ 
    &\leq \sqrt{\frac{\exp\left(4/\sigma\right)-1}{\left(\exp\left(1/\sigma\right)+1\right)^4}}\cdot \frac{d_{\max}^2}{2} \cdot \left(\sum_{i = 1}^A \sqrt{2\mathcal{p}_i(1-\mathcal{p}_i)}\right). \nonumber 
\end{align}
This is an upper bound for $\E{D_{\mathcal{d}_{\max}}}$; using the sandwich argument~\ref{sec:sandwich}, the same upper bound applies to $\E{D}$. Let $N \to \infty$ and $\sigma \to \infty$, we see that the right hand side of~\eqref{eq:indepschemebound} behave as $1/\sqrt{\sigma}$; hence we conclude that as $\sigma \to \infty$, $(\alpha,\beta)$ is asymptotically achievable for $\alpha = -\infty, \beta > -1/2$.
\end{proof}

\subsubsection{Random Permutation}

Consider the following scheme: divide the indices into windows of size $k$, i.e., fix some non-negative integer $C \leq k-1$, and define $w_0 = \{-C,-C+1,\ldots,0,\ldots,k-C-1\}, w_1 = \{k-C,k-C+1,\ldots,2k-C-1\}, w_{-1} = \{-C-k,-C-k+1,\ldots,-C-1\}$, etc. Let $w_{r}$ be the window of remainders, i.e., $w_r$ contains all indices on both ends that do not fit into any of the windows of length $k$ above. We assume that $k$ grows sub-linear to $\sigma$. Within each window $w_i$ that is not the window of remainders, $\hat{X}_j, j \in w_i$ is a random permutation of $X_j, j \in w_i$. In other words, we repeatedly apply the random permutation test channel~\cite{tang2023capacity} on each window. For the window of remainders, the reconstructions always output the first symbol.  
% Since we assume $N \to \infty$, W.L.O.G.\ assume that no source symbol is included in more than one window. 

$k = 1$ marks a special case: the rate is $\log_2 A$ since we need to transmit the exact symbols, and the distortion is $0$ regardless of the choice of $\sigma$ since the reconstruction is identical to the source. In this case, $(\alpha,\beta)$ is asymptotically achievable if $\alpha = 0$ and $\beta = -\infty$.

\begin{thm}
	The pair $(\alpha,\beta)$ is asymptotically achievable if $\alpha + \beta > -3/2$ and $-1 < \alpha < 0$.
\end{thm}

\begin{proof}
For fixed $\sigma$, choose $k = [\sigma^\gamma]$ for some $0 < \gamma < 1$. Choose the sequence of $N$ such that $\frac{k}{2N+1}\floor{\frac{2N+1}{k}} \to 1$ as $N \to \infty$. We use $I$ to denote the set of indices for the length $k$ windows which are contained completely within the $\{-N,\ldots,N\}$ block, with $|I| = \floor{(2N+1)/k}$.

For any finite $N$, for each window $w_i$ that is not the window of remainders, we need to transmit the count of each but the last symbol $i = 1,2,\ldots,A-1$ with $\ceil{\log_2 (k+1)}$ bits. For the window of remainders $w_r$, we do not need to transmit anything. Hence the rate is 
\begin{equation}
	(A-1)\ceil{\log_2 (k+1)} \floor{\frac{2N+1}{k}} \frac{1}{2N+1}.
\end{equation}
By our assumption, $\floor{\frac{2N+1}{k}} \frac{1}{2N+1} \to \frac{1}{k}$ as $N \to \infty$; hence we conclude that the asymptotic rate is $(A-1)\ceil{\log_2 (k+1)}/k$. 

For the distortion, write $\E{D_{\mathcal{d}_{\max}}}$ as
\begin{align}
	\E{D_{\mathcal{d}_{\max}}} &= \frac{1}{2N+1} \sum_{n = -N}^N \E{W_2^2(Y_n,\hat{Y}_n)} \\
	&= \frac{1}{2N+1} \left( \sum_{j \in I} \sum_{n \in w_j} \E{W_2^2(Y_n,\hat{Y}_n)} \right. \nonumber \\
	&\phantom{=} \left. \ + \sum_{n \in w_r} \E{W_2^2(Y_n,\hat{Y}_n)} \right) \\
	&\leq \left( \floor{\frac{2N+1}{k}} \max_{j \in I} k \max_{n \in w_j} \E{W_2^2(Y_n,\hat{Y}_n)} \right. \nonumber \\
	&\phantom{=} \left. \ + \left( (2N+1)-\floor{\frac{2N+1}{k}}k \right) \right)/(2N+1) \label{eq:boundby1}\\
	&= \frac{k}{2N+1}\floor{\frac{2N+1}{k}} \max_{j \in I} \max_{n \in w_j} \E{W_2^2(Y_n,\hat{Y}_n)} \nonumber \\
	&\phantom{=} \ + 1 - \frac{k}{2N+1}\floor{\frac{2N+1}{k}},
\end{align}
where~\eqref{eq:boundby1} holds because $\E{W_2^2(Y_n,\hat{Y}_n)} \leq 1$ for all $n$.

Fix a $j \in I$ and $n \in w_j$. We start with the expectation 
\begin{align}
	\E{W_2^2(Y_n,\hat{Y}_n)} &= \frac{d^2_{\max}}{2} \cdot \left( \sum_{i = 1}^A \mathbb{E} \left[ \left| \sum_{m = -\infty}^\infty q_\sigma(m-n) \mathbbm{1}_{X_m = i} \right. \right. \right. \nonumber \\
    &\phantom{=} \left.\left.\left. - \sum_{m = -\infty}^\infty q_\sigma(m-n) \mathbbm{1}_{\hat{X}_m = i} \right| \right] \right).
\end{align}
For each $i = 1,2,\ldots,A$, consider the summand term
\begin{align}
	&\phantom{=} \E{\abs{\sum_{m} q_\sigma(m-n) \mathbbm{1}_{X_m = i} - \sum_{m} q_\sigma(m-n) \mathbbm{1}_{\hat{X}_m = i}}} \\
	&= \E{\abs{\sum_{m = -\infty}^\infty q_\sigma(m-n) \left(\mathbbm{1}_{X_m = i} - \mathbbm{1}_{\hat{X}_m = i}\right)}} \\
	&= \E{\abs{\sum_{\ell' = -\infty}^\infty \sum_{m \in w_{\ell'}} q_\sigma(m-n) \left(\mathbbm{1}_{X_m = i} - \mathbbm{1}_{\hat{X}_m = i}\right)}} \\
	&= \frac{\exp\left(1/\sigma\right)-1}{\exp\left(1/\sigma\right)+1} \times \\
    &\phantom{=}\ \E{\abs{\sum_{\ell' = -\infty}^\infty \sum_{m \in w_{\ell'}} \exp\left(-\frac{\abs{m-n}}{\sigma}\right) \left(\mathbbm{1}_{X_m = i} - \mathbbm{1}_{\hat{X}_m = i}\right)}}. \nonumber 
 \end{align}
Break the summation into $\ell' > j, \ell' < j$ and $\ell' = j$, and let $\ell = \ell' - j$. Let $C'$ denote the number of indices to the left of $n$ in the same window $w_j$, i.e., $C' = (n+C)\mod k$. We can further write
\begin{align}
	&= \frac{\exp\left(1/\sigma\right)-1}{\exp\left(1/\sigma\right)+1} \times \nonumber \\
    &\phantom{=}\ \mathbb{E} \left[ \left| \sum_{\ell = 1}^\infty \exp\left(\frac{C'}{\sigma}-\frac{\ell k}{\sigma}\right) \right. \right. \nonumber \\
    &\phantom{= \mathbb{E}} \quad \left( \sum_{m = 0}^{k-1} \exp\left(-\frac{\abs{m}}{\sigma}\right) \left(\mathbbm{1}_{X_{\ell k-C'+m} = i} - \mathbbm{1}_{\hat{X}_{\ell k-C'+m} = i}\right) \right) \nonumber \\
	&\phantom{= \mathbb{E}} \ + \sum_{\ell = 1}^\infty \exp\left(\frac{k-C'-1}{\sigma}-\frac{\ell k}{\sigma}\right) \nonumber \\
    &\phantom{= \mathbb{E}} \quad \left( \sum_{m = 0}^{k-1} \exp\left(-\frac{\abs{m}}{\sigma}\right) \left(\mathbbm{1}_{X_{-C'-1-(\ell-1)k-m} = i} \right. \right. \nonumber \\
    &\phantom{= \mathbb{E}} \left. \left. \quad - \mathbbm{1}_{\hat{X}_{-C'-1-(\ell-1)k-m} = i}\right) \right) \nonumber \\
	&\phantom{= \mathbb{E}} \left. \left. \ + \sum_{m = -C'}^{k-C'-1} \exp\left(-\frac{\abs{m}}{\sigma}\right) \left(\mathbbm{1}_{X_m = i} - \mathbbm{1}_{\hat{X}_m = i}\right) \right|\right]. \label{eq:windowssum}
\end{align}
For $\ell = 1,2,\ldots$, define
\begin{align}
	\mathcal{W}_\ell &= \sum_{m = 0}^{k-1} \exp\left(-\frac{\abs{m}}{\sigma}\right) \left(\mathbbm{1}_{X_{\ell k-C'+m} = i} - \mathbbm{1}_{\hat{X}_{\ell k-C'+m} = i}\right), \\
	\mathcal{W}_{-\ell} &= \sum_{m = 0}^{k-1} \exp\left(-\frac{\abs{m}}{\sigma}\right) \left(\mathbbm{1}_{X_{-C'-1-(\ell-1)k-m} = i} \right. \nonumber \\
    &\phantom{-} \left.- \mathbbm{1}_{\hat{X}_{-C'-1-(\ell-1)k-m} = i}\right),
\end{align}
and define
\begin{equation}
	\mathcal{W}_0 = \sum_{m = -C'}^{k-C'-1} \exp\left(-\frac{\abs{m}}{\sigma}\right) \left(\mathbbm{1}_{X_m = i} - \mathbbm{1}_{\hat{X}_m = i}\right).
\end{equation}
Notice that $\{\mathcal{W}_j\}_{j = -\infty}^\infty$ are independent, and $\{\mathcal{W}_j\}_{j \neq 0}$ are i.i.d. 

Since by definition,
\begin{equation}
	\sum_{m = 0}^{k-1} \mathbbm{1}_{X_{\ell k-C'+m} = i} - \mathbbm{1}_{\hat{X}_{\ell k-C'+m} = i} = 0,
\end{equation}
we see that for $\ell = 1,2,\ldots$, for some $-|m|/\sigma \leq \xi_m \leq 0$ for each $m = 0,1,\ldots,k-1$,
\begin{align}
	\mathcal{W}_\ell &= \sum_{m = 0}^{k-1} \left(\exp\left(-\frac{\abs{m}}{\sigma} \right) - 1 \right)  \times \nonumber \\
	&\phantom{=} \quad \left(\mathbbm{1}_{X_{\ell k-C'+m} = i} - \mathbbm{1}_{\hat{X}_{\ell k-C'+m} = i}\right) \\
	&= \sum_{m = 0}^{k-1} \left(-\frac{\abs{m}}{\sigma} + \frac{\exp\left(\xi_m\right)}{2}\frac{\abs{m}^2}{\sigma^2}\right) \times \nonumber \\
	&\phantom{=} \quad \left(\mathbbm{1}_{X_{\ell k-C'+m} = i} - \mathbbm{1}_{\hat{X}_{\ell k-C'+m} = i}\right).
\end{align}
Similarly,
\begin{align}
	\mathcal{W}_{-\ell} &= \sum_{m = 0}^{k-1} \left(-\frac{\abs{m}}{\sigma} + \frac{\exp\left(\xi_m\right)}{2}\frac{\abs{m}^2}{\sigma^2}\right) \times \nonumber \\
	&\phantom{=} \quad \left(\mathbbm{1}_{X_{-C'-1-(\ell-1)k-m} = i} - \mathbbm{1}_{\hat{X}_{-C'-1-(\ell-1)k-m} = i}\right).
\end{align}
And for $\ell = 0$,
\begin{align}
	\mathcal{W}_{0} &= \sum_{m = -C'}^{k-C'-1} \left(-\frac{\abs{m}}{\sigma} + \frac{\exp\left(\xi_m\right)}{2}\frac{\abs{m}^2}{\sigma^2}\right) \times \nonumber \\
	&\phantom{=} \quad \left(\mathbbm{1}_{X_m = i} - \mathbbm{1}_{\hat{X}_m = i}\right).
\end{align}

For any integer $j$, a direct calculation reveals that all summands within $w_j$ are negatively correlated. Also, $\mathrm{Var}\left(\mathbbm{1}_{X_m = i} - \mathbbm{1}_{\hat{X}_m = i}\right) \leq 1$ for all integers $m$. Thus, for all $j = \pm1,\pm2,\ldots$,
\begin{align}
	\mathrm{Var} \left( \mathcal{W}_j\right)  &\leq \sum_{m = 0}^{k-1} \left(\frac{\abs{m}}{\sigma} - \frac{\exp\left(\xi_m\right)}{2}\frac{\abs{m}^2}{\sigma^2}\right)^2 \\
	&\leq \frac{k^3}{\sigma^2} \left( 1 - \frac{k}{\sigma} + \frac{k^2}{\sigma^2} \right). \label{eq:varbound}
\end{align}
For $j = 0$, a similar bound holds:
\begin{align}
	\mathrm{Var} \left( \mathcal{W}_0\right) &\leq 2 \sum_{m = 0}^{k-1} \left(\frac{\abs{m}}{\sigma} - \frac{\exp\left(\xi_m\right)}{2}\frac{\abs{m}^2}{\sigma^2}\right)^2 \\
	&\leq 2\frac{k^3}{\sigma^2} \left( 1 - \frac{k}{\sigma} + \frac{k^2}{\sigma^2} \right).
\end{align}

We now have
\begin{align}
	&\phantom{=} \E{\abs{\sum_{m} q_\sigma(m-n) \mathbbm{1}_{X_m = i} - \sum_{m} q_\sigma(m-n) \mathbbm{1}_{\hat{X}_m = i}}} \\
	&= \frac{\exp\left(1/\sigma\right)-1}{\exp\left(1/\sigma\right)+1} \mathbb{E} \left[ \left| \sum_{\ell = 1}^\infty \exp\left(\frac{C'}{\sigma}-\frac{\ell k}{\sigma}\right) \mathcal{W}_\ell \right. \right. \nonumber \\
	&\phantom{= \mathbb{E}} \left. \left.  + \sum_{\ell = 1}^\infty \exp\left(\frac{k-C'-1}{\sigma}-\frac{\ell k}{\sigma}\right) \mathcal{W}_{-\ell} + \mathcal{W}_0 \right|\right] \\
	&\leq \frac{\exp\left(1/\sigma\right)-1}{\exp\left(1/\sigma\right)+1} \left(\sum_{\ell = 1}^\infty \exp\left(\frac{2C'}{\sigma}-\frac{2\ell k}{\sigma}\right) \mathrm{Var}(\mathcal{W}_\ell)\right. \nonumber \\
	&\phantom{= \mathbb{E}} \left. + \sum_{\ell = 1}^\infty \exp\left(\frac{k-C'-1}{\sigma}-\frac{\ell k}{\sigma}\right) \mathrm{Var}(\mathcal{W}_{-\ell}) + \mathrm{Var}(\mathcal{W}_0) \right)^{1/2} \\
	&\leq \frac{\exp\left(1/\sigma\right)-1}{\exp\left(1/\sigma\right)+1} \left[ \left(\sum_{\ell = 1}^\infty \left(\exp\left(\frac{2C'}{\sigma}-\frac{2\ell k}{\sigma}\right) \right. \right. \right. \nonumber \\
	&\phantom{= \mathbb{E}} \left. \left. \left. + \exp\left(\frac{2(k-C'-1)}{\sigma}-\frac{2\ell k}{\sigma}\right) \right)+ 2 \right) \times \right. \nonumber \\
	&\phantom{= \mathbb{E}} \left. \frac{k^3}{\sigma^2} \left( 1 - \frac{k}{\sigma} + \frac{k^2}{\sigma^2} \right) \right]^{1/2} \\
	&= \frac{\exp\left(1/\sigma\right)-1}{\exp\left(1/\sigma\right)+1} \left[ \left( \left(\exp\left(\frac{2C'}{\sigma}\right) + \exp\left(\frac{2(k-C'-1)}{\sigma}\right)\right) \right. \right. \nonumber \\
	&\phantom{= \mathbb{E}} \left. \left. \frac{1}{e^{2k/\sigma}-1} + 2 \right) \frac{k^3}{\sigma^2} \left( 1 - \frac{k}{\sigma} + \frac{k^2}{\sigma^2} \right) \right]^{1/2}.
\end{align}

Maximizing over $n \in w_j$, which is equivalent to maximizing over $0 \leq C' \leq k-1$, we see that
\begin{align}
	&\phantom{=} \max_{n \in w_j} \E{\abs{\sum_{m} q_\sigma(m-n) \mathbbm{1}_{X_m = i} - \sum_{m} q_\sigma(m-n) \mathbbm{1}_{\hat{X}_m = i}}} \\
	&\leq \max_{0 \leq C' \leq k-1} \frac{\exp\left(1/\sigma\right)-1}{\exp\left(1/\sigma\right)+1} \times \nonumber \\ 
	&\phantom{\leq} \ \left[ \left( \frac{e^{2C'/\sigma} + e^{2(k-C'-1)/\sigma}}{e^{2k/\sigma}-1} + 2 \right) \frac{k^3}{\sigma^2} \left( 1 - \frac{k}{\sigma} + \frac{k^2}{\sigma^2} \right) \right]^{1/2} \\
	&= \frac{\exp\left(1/\sigma\right)-1}{\exp\left(1/\sigma\right)+1} \times \nonumber \\
	&\phantom{\leq} \ \left[ \left( \frac{1 + e^{2(k-1)/\sigma}}{e^{2k/\sigma}-1} + 2 \right) \frac{k^3}{\sigma^2} \left( 1 - \frac{k}{\sigma} + \frac{k^2}{\sigma^2} \right) \right]^{1/2}. \label{eq:boundforwj}
\end{align}

We see the bound~\eqref{eq:boundforwj} holds for all $j \in I$. We conclude that
\begin{align}
	&\phantom{=} \lim_{N \to \infty} \E{D_{\mathcal{d}_{\max}}} \nonumber \\
	&\leq \lim_{N \to \infty} \frac{k}{2N+1}\floor{\frac{2N+1}{k}} \max_{j \in I} \max_{n \in w_j} \E{W_2^2(Y_n,\hat{Y}_n)} \nonumber \\
	&\phantom{=} \ + 1 - \lim_{N \to \infty} \frac{k}{2N+1}\floor{\frac{2N+1}{k}} \\
	&= \max_{j \in I} \max_{n \in w_j} \E{W_2^2(Y_n,\hat{Y}_n)} \\
	&\leq \frac{\exp\left(1/\sigma\right)-1}{\exp\left(1/\sigma\right)+1} \times \nonumber \\
	&\phantom{\leq} \ \left[ \left( \frac{1 + e^{2(k-1)/\sigma}}{e^{2k/\sigma}-1} + 2 \right) \frac{k^3}{\sigma^2} \left( 1 - \frac{k}{\sigma} + \frac{k^2}{\sigma^2} \right) \right]^{1/2}. \label{eq:rdmpermutebound}
\end{align}

Similar to last subsection, this is an upper bound for $\E{D_{\mathcal{d}_{\max}}}$; using the sandwich argument~\ref{sec:sandwich}, the same upper bound applies to $\E{D}$. Recall that $k = \left[\sigma^\gamma\right]$ for some $0 < \gamma < 1$, and let $\sigma \to \infty$, the right hand side of~\eqref{eq:rdmpermutebound} goes to $0$ as $1/\sigma^{3/2-\gamma}$. We conclude that $(\alpha,\beta)$ is asymptotically achievable if $-1 < \alpha < 0, \alpha + \beta > -3/2$.
\end{proof}

\subsection{Converse -- Lower Bound}

Denote $W_{2,\mathcal{d}_{\min}}^2$ by $W_2^2$ in this subsection.

\begin{thm}
%For $\epsilon > 0$, $\E{D} \geq O\left(1/\sqrt{\sigma}\right)$; i.e. as $\sigma \to \infty$,
%\begin{align}
%& R < \frac{1}{\sigma^2} \\ \nonumber
%\implies & \E{D} \geq O\left(1/\sqrt{\sigma}\right).
%\end{align}
	The pair $(\alpha,\beta)$ is not asymptotically achievable if $\alpha < -2$ and $\beta = -1/2$.
\end{thm}

\begin{proof}
We prove that, for $\sigma$ large enough, if the rate $R \leq 1/\sigma^{2+\epsilon}$, then the expected distortion $\E{D} \geq O\left(1/\sqrt{\sigma}\right)$. Consider the same partition as in the permutation scheme, with length $k$ windows $w_0,w_1,w_{-1}, \ldots$, and the window of remainders $w_r$. Fix $\sigma$, let $k = \floor{\sigma^{1+\eta}}$ for some positive $\eta < \epsilon$, and choose the sequence of $N$ such that $\frac{k}{2N+1}\floor{\frac{2N+1}{k}} \to 1$ as $N \to \infty$. We use $I$ to denote the set of indices for the length $k$ windows that reside entirely within the $\{-N,\ldots,N\}$ block, with $|I| = \floor{(2N+1)/k}$. We can write
\begin{align}
	\E{D_{\mathcal{d}_{\min}}} &= \frac{1}{2N+1} \sum_{n = -N}^N \E{W_2^2(Y_n,\hat{Y}_n)} \\
	&= \frac{1}{2N+1} \left( \sum_{j \in I} \sum_{n \in w_j} \E{W_2^2(Y_n,\hat{Y}_n)} \right. \nonumber \\
	&\phantom{=} \left. + \sum_{n \in w_r} \E{W_2^2(Y_n,\hat{Y}_n)} \right) \\
	&\geq \frac{1}{2N+1} \sum_{j \in I} \sum_{n \in W_j} \E{W_2^2(Y_n,\hat{Y}_n)},
\end{align}
where $W_j$ contains all indices in window $w_j$ that are at least $k/\sigma^{\eta/2}$ away from the boundary of the window. In other words, we neglect an $1/\sigma^{\eta/2}$ fraction of indices to both ends for each window. Notice that this proportion vanishes as $\sigma \to \infty$.

Let $X_j^k$ (resp., $\hat{X}_j^k$) denote all variables within window $w_j$, and $x_j^k$ (resp., $\hat{x}_j^k$) denote a particular realization of them. Define
\begin{equation}
	I_0 = \left\{j \in I: I(X_j^k;\hat{X}_j^k) \leq \frac{1}{\sigma^{1+(\epsilon-\eta)/2}} \right\},
\end{equation}
i.e., $I_0$ is the subset of $I$ which contains all windows such that the source and reconstruction restricted to that window have low mutual information. We can further bound
\begin{align}
	&\phantom{\geq} \E{D_{\mathcal{d}_{\min}}} \geq \frac{1}{2N+1} \sum_{j \in I_0} \sum_{n \in W_j} \E{W_2^2(Y_n,\hat{Y}_n)} \\
	&\geq \frac{1}{2N+1} \sum_{j \in I_0} \floor{k-\frac{2k}{\sigma^{\eta/2}}}\cdot \min_{n \in W_j} \E{W_2^2(Y_n,\hat{Y}_n)} \\
	&\geq \frac{\floor{k-\frac{2k}{\sigma^{\eta/2}}}}{2N+1} \cdot |I_0| \cdot \min_{j \in I_0} \min_{n \in W_j} \E{W_2^2(Y_n,\hat{Y}_n)}.
\end{align}

%	&\geq \frac{k}{2N+1} \floor{\frac{2N+1}{k}} \quad \frac{1}{\floor{(2N+1)/k}} \times \nonumber \\
%	&\phantom{\geq} \sum_{j \in I} \frac{1}{k} \sum_{n \in W_j} \min_{n \in W_j} \E{W_2^2(Y_n,\hat{Y}_n)} \\
%	&= \frac{k}{2N+1} \floor{\frac{2N+1}{k}} \frac{1}{\floor{(2N+1)/k}} \times \nonumber \\
%	&\phantom{\geq} \quad \sum_{j \in I} \frac{\floor{k(1-2/\sigma^{\eta/2})}}{k} \min_{n \in W_j} \E{W_2^2(Y_n,\hat{Y}_n)},

Suppose $j \in I_0$ and $n \in W_j$. Then,
\begin{align}
	&\phantom{=} \E{W_2^2(Y_n,\hat{Y}_n)} \nonumber \\
    &= \frac{d^2_{\min}}{2}\cdot \sum_{i = 1}^A \mathbb{E} \left[ \left| \sum_{m = -\infty}^\infty q_\sigma(m-n) \mathbbm{1}_{X_m = i} \right. \right. \nonumber \\
	&\phantom{=} \left. \left. - \sum_{m = -\infty}^\infty q_\sigma(m-n) \mathbbm{1}_{\hat{X}_m = i} \right| \right] \\
	&= \frac{d^2_{\min}}{2}\cdot \sum_{i = 1}^A \mathbb{E} \left[ \left| \sum_{m  \in w_j} q_\sigma(m-n) \left(\mathbbm{1}_{X_m = i} - \mathbbm{1}_{\hat{X}_m = i} \right) \right.\right. \nonumber \\
    &\phantom{=} \left.\left. + \sum_{m  \notin w_j} q_\sigma(m-n) \left(\mathbbm{1}_{X_m = i} - \mathbbm{1}_{\hat{X}_m = i} \right) \right| \right] \\
	&\geq \frac{d^2_{\min}}{2}\cdot \sum_{i = 1}^A \left(\E{\abs{\sum_{m  \in w_j} q_\sigma(m-n) \left(\mathbbm{1}_{X_m = i} - \mathbbm{1}_{\hat{X}_m = i}\right) }} \right. \nonumber \\
	&\quad \left. - \E{\abs{\sum_{m \notin w_j} q_\sigma(m-n) \left(\mathbbm{1}_{X_m = i}- \mathbbm{1}_{\hat{X}_m = i}\right)}}\right). \label{eq:converseexp}
\end{align}

Consider the first term in~\eqref{eq:converseexp}, and fix an $i = 1,2,\ldots,A$. We can write
\begin{align}
	&\phantom{=} \E{\abs{\sum_{m  \in w_j} q_\sigma(m-n) \left(\mathbbm{1}_{X_m = i} - \mathbbm{1}_{\hat{X}_m = i}\right) }} \nonumber \\
    &= \sum_{x_j^k,\hat{x}_j^k} d\left(x_j^k,\hat{x}_j^k\right) p\left(x_j^k,\hat{x}_j^k\right), \\
\intertext{where $d\left(x_j^k,\hat{x}_j^k\right) = \abs{\sum_{m \in w_j} q_\sigma(m-n) \left(\mathbbm{1}_{x_m = i} - \mathbbm{1}_{\hat{x}_m = i}\right)}$, and $p(x_j^k,\hat{x}_j^k)$ is the joint distribution between the source and reconstruction over $w_j$. We can further write}
	&= \sum_{x_j^k,\hat{x}_j^k} d\left(x_j^k,\hat{x}_j^k\right) \left[p\left(x_j^k,\hat{x}_j^k\right) \right. \\
    &\phantom{=} \left. + p\left(x_j^k\right)p\left(\hat{x}_j^k\right) - p\left(x_j^k\right)p\left(\hat{x}_j^k\right) \right] \nonumber \\
	&= \sum_{x_j^k,\hat{x}_j^k} d\left(x_j^k,\hat{x}_j^k\right) p\left(x_j^k\right)p\left(\hat{x}_j^k\right) \\ 
	&\quad + \sum_{x_j^k,\hat{x}_j^k} d\left(x_j^k,\hat{x}_j^k\right) \left(p\left(x_j^k,\hat{x}_j^k\right) - p\left(x_j^k\right)p\left(\hat{x}_j^k\right) \right) \nonumber \\
	&\geq \sum_{x_j^k,\hat{x}_j^k} d\left(x_j^k,\hat{x}_j^k\right) p\left(x_j^k\right)p\left(\hat{x}_j^k\right) \label{eq:lessthanone} \\
	&\quad - \sum_{x_j^k,\hat{x}_j^k} \abs{p\left(x_j^k,\hat{x}_j^k\right) - p\left(x_j^k\right)p\left(\hat{x}_j^k\right)} \nonumber \\
	&= \EE{\Perp}{\abs{\sum_{m \in w_j} q_\sigma(m-n) \left(\mathbbm{1}_{X_m = i} - \mathbbm{1}_{\hat{X}_m = i}\right)}} \label{eq:centerwindowexp} \\
	&\quad - \sum_{x_j^k,\hat{x}_j^k} \abs{p\left(x_j^k,\hat{x}_j^k\right) - p\left(x_j^k\right)p\left(\hat{x}_j^k\right)}, \nonumber
\end{align}
where $\EE{\Perp}{\cdot}$ denotes the expectation assuming $X_j^k$ and $\hat{X}_j^k$ are independent of each other, and \eqref{eq:lessthanone} used the fact that $d\left(x_j^k,\hat{x}_j^k\right) \leq 1$ for all realizations. We see that 
\begin{align}
	&\phantom{\geq} \min_{j \in I_0} \min_{n \in W_j} \E{W_2^2(Y_n,\hat{Y}_n)} \nonumber \\
	&\geq \min_{j \in I_0} \min_{n \in W_j} \frac{d^2_{\min}}{2} \cdot \sum_{i = 1}^A \\
	&\ \left( \EE{\Perp}{\abs{\sum_{m \in w_j} q_\sigma(m-n) \left(\mathbbm{1}_{X_m = i} - \mathbbm{1}_{\hat{X}_m = i}\right)}} \right. \nonumber \\ 
	&\ - \sum_{x_j^k,\hat{x}_j^k} \abs{p\left(x_j^k,\hat{x}_j^k\right) - p\left(x_j^k\right)p\left(\hat{x}_j^k\right)} \nonumber \\
	&\ \left. - \E{\abs{\sum_{m \notin w_j} q_\sigma(m-n) \left(\mathbbm{1}_{X_m = i}- \mathbbm{1}_{\hat{X}_m = i}\right)}}\right) \nonumber \\
	&\geq \frac{d^2_{\min}}{2} \cdot \sum_{i = 1}^A \label{eq:sumwindowexp} \\
	&\ \left( \min_{j \in I_0} \min_{n \in W_j} \EE{\Perp}{\abs{\sum_{m \in w_j} q_\sigma(m-n) \left(\mathbbm{1}_{X_m = i} - \mathbbm{1}_{\hat{X}_m = i}\right)}} \right. \nonumber \\ 
	&\ - \max_{j \in I_0} \max_{n \in W_j} \sum_{x_j^k,\hat{x}_j^k} \abs{p\left(x_j^k,\hat{x}_j^k\right) - p\left(x_j^k\right)p\left(\hat{x}_j^k\right)} \nonumber \\
	&\ \left. - \max_{j \in I_0} \max_{n \in W_j} \E{\abs{\sum_{m \notin w_j} q_\sigma(m-n) \left(\mathbbm{1}_{X_m = i}- \mathbbm{1}_{\hat{X}_m = i}\right)}}\right). \nonumber 
\end{align}

Now, fix an $i = 1,2,\ldots,A$, and consider the three terms in~\eqref{eq:sumwindowexp}:
\begin{align}
	&\phantom{-}\min_{j \in I_0} \min_{n \in W_j} \EE{\Perp}{\abs{\sum_{m \in w_j} q_\sigma(m-n) \left(\mathbbm{1}_{X_m = i} - \mathbbm{1}_{\hat{X}_m = i}\right)}} \nonumber \\ 
	&- \max_{j \in I_0} \max_{n \in W_j} \sum_{x_j^k,\hat{x}_j^k} \abs{p\left(x_j^k,\hat{x}_j^k\right) - p\left(x_j^k\right)p\left(\hat{x}_j^k\right)} \label{eq:windowexp} \\
	&- \max_{j \in I_0} \max_{n \in W_j} \E{\abs{\sum_{m \notin w_j} q_\sigma(m-n) \left(\mathbbm{1}_{X_m = i}- \mathbbm{1}_{\hat{X}_m = i}\right)}}. \nonumber
\end{align}

For the first term in~\eqref{eq:windowexp}, define
\begin{align}
U_{ijn} &= \sum_{m  \in w_j} q_\sigma(m-n) \mathbbm{1}_{X_m = i}, \\ 
\hat{U}_{ijn} &= \sum_{m \in w_j} q_\sigma(m-n) \mathbbm{1}_{\hat{X}_m = i}, \\
Q_{ijn} &= \sqrt{\mathcal{p}_i\left(1-\mathcal{p}_i\right) \cdot \sum_{m \in w_j} q_\sigma(m-n)^2}.
\end{align}
We can bound the expectation by
\begin{align}
	&\phantom{=} \EE{\Perp}{\abs{\sum_{m \in w_j} q_\sigma(m-n) \left(\mathbbm{1}_{X_m = i} - \mathbbm{1}_{\hat{X}_m = i}\right)}} \\
	&= \EE{\Perp}{\abs{U_{ijn} - \hat{U}_{ijn}}} \\
	&\geq \EE{\Perp}{\abs{U_{ijn} - \hat{U}_{ijn}} \cdot \mathbbm{1}_{U_{ijn} > \E{U_{ijn}} + Q_{ijn} \wedge \hat{U}_{ijn} \leq \E{U_{ijn}}}} \nonumber \\
	&\phantom{\geq} \ + \EE{\Perp}{\abs{U_{ijn} - \hat{U}_{ijn}} \cdot \mathbbm{1}_{U_{ijn} \leq \E{U_{ijn}} - Q_{ijn} \wedge \hat{U}_{ijn} > \E{U_{ijn}}}} \\
	&\geq Q_{ijn} \cdot \mathbb{P}\left( U_{ijn} > \E{U_{ijn}} + Q_{ijn}\right) \cdot \mathbb{P}\left(\hat{U}_{ijn} \leq \E{U_{ijn}}\right) \nonumber \\
	&\phantom{\geq} \ + Q_{ijn} \cdot \mathbb{P}\left( U_{ijn} \leq \E{U_{ijn}} - Q_{ijn} \right) \cdot \mathbb{P}\left(\hat{U}_{ijn} > \E{U_{ijn}} \right). \label{eq:sourcereconindep}
\end{align}
We can further bound $Q_{ijn}$ by
\begin{align}
	Q_{ijn} &\geq \sqrt{p_0 (1-p_0) \sum_{m = 0}^{k-1} q_\sigma(m)^2} \\
	&= \sqrt{p_0 (1-p_0) (e^{1/\sigma}-1) \frac{e^{2/\sigma} \left( 1 - e^{-2k/\sigma} \right)}{\left( e^{1/\sigma}+1 \right)^3}},
\end{align}
where $p_0 = \arg\min_{\mathcal{p}_i, i = 1,2,\ldots,A} \mathcal{p}_i (1 - \mathcal{p}_i)$. Now,
\begin{align}
	&\phantom{\geq} \min_{j \in I_0} \min_{n \in W_j} \EE{\Perp}{\abs{\sum_{m \in w_j} q_\sigma(m-n) \left(\mathbbm{1}_{X_m = i} - \mathbbm{1}_{\hat{X}_m = i}\right)}} \nonumber \\
	&\geq \sqrt{p_0 (1-p_0) (e^{1/\sigma}-1) \frac{e^{2/\sigma} \left( 1 - e^{-2k/\sigma} \right)}{\left( e^{1/\sigma}+1 \right)^3}} \times \nonumber \\
	&\phantom{\geq} \min_{j \in I_0} \min_{n \in W_j} \left( \mathbb{P}\left( U_{ijn} > \E{U_{ijn}} + Q_{ijn}\right) \cdot \mathbb{P}\left(\hat{U}_{ijn} \leq \E{U_{ijn}}\right) \right. \nonumber \\
	&\phantom{\geq} \left. + \mathbb{P}\left( U_{ijn} \leq \E{U_{ijn}} - Q_{ijn} \right) \cdot \mathbb{P}\left(\hat{U}_{ijn} > \E{U_{ijn}} \right) \right).
\end{align}
Consider the following variables
\begin{align}
	&\sum_{m  \in w_{j_k}} \frac{q_{\sigma}(m-n_k)}{\sqrt{\sum_{m \in w_{j_k}} q_{\sigma}(m-n_k)^2}} \left(\mathbbm{1}_{X_m = i} - \mathcal{p}_i\right) \\
	\text{and } &\sum_{m  \in w_{j_{k'}}} \frac{q_{\sigma}(m-n_{k'})}{\sqrt{\sum_{m \in w_{j_{k'}}} q_{\sigma}(m-n_{k'})^2}} \left(\mathbbm{1}_{X_m = i} - \mathcal{p}_i\right),
\end{align}
where for each $k$ (resp., $k'$), $j_k$ and $n_k$ (resp., $j_{k'}$ and $n_{k'}$) are the minimizer of $\mathbb{P}\left( U_{ijn} > \E{U_{ijn}} + Q_{ijn}\right)$ (resp., $\mathbb{P}\left( U_{ijn} \leq \E{U_{ijn}} - Q_{ijn} \right)$), i.e., for a fixed $k$, 
\begin{align}
	&\phantom{\leq} \mathbb{P}\left( U_{ij_kn_k} > \E{U_{ij_kn_k}} + Q_{ij_kn_k}\right) \\
	&\leq \mathbb{P}\left( U_{ijn} > \E{U_{ijn}} + Q_{ijn}\right) \text{ for all } j \in I_0, n \in W_j, \nonumber
\end{align}
and similarly for $k'$. By Lyapunov's Central Limit Theorem (CLT) \cite[Theorem 27.3]{billingsley2017probability},
\begin{align}
	&\phantom{\to} \sum_{m \in w_{j_k}} \frac{q_\sigma(m-n_k)}{\sqrt{\sum_{m \in w_{j_k}} q_\sigma(m-n_k)^2}} \left(\mathbbm{1}_{X_m = i} - \mathcal{p}_i\right) \nonumber \\
    &\overset{D}{\to} \mathcal{N}\left(0,\mathcal{p}_i\left(1-\mathcal{p}_i\right) \right) \text{ as } k \to \infty. \label{eq:lyapunov}
\end{align}
To validate the condition for Lyapunov's CLT~\cite[(27.16)]{billingsley2017probability}, choose $\delta = 2$. Recall that $k = \floor{\sigma^{1+\eta}}$. The sum of variances $s_k^2 = \mathcal{p}_i(1 - \mathcal{p}_i)$ for all $k$, so we only need to check $\sum_{m \in w_j} q_\sigma(m-n_k)^4/(\sum_{m \in w_j} q_\sigma(m-n_k)^2)^2 \to 0$ as $k \to \infty$. By direct computation, 
\begin{equation}
	\frac{\sum_{m \in w_j} q_\sigma(m-n_k)^4}{(\sum_{m \in w_j} q_\sigma(m-n_k)^2)^2} \leq \left(e^{1/\sigma} - 1\right).
\end{equation}
Thus we see that the ratio goes to $0$ as $k \to \infty$, i.e., the condition~\cite[(27.16)]{billingsley2017probability} is verified. Lyapunov's CLT holds for the variable with $k'$ via the same argument.~\eqref{eq:lyapunov} implies that for some constant $c' > 0$,
\begin{align}
	&\mathbb{P}\left( \abs{ \sum_{m \in w_{j_k}} \frac{q_\sigma(m-n_k)}{\sqrt{\sum_{m \in w_{j_k}} q_\sigma(m-n_k)^2}} \left(\mathbbm{1}_{X_m = i} - \mathcal{p}_i\right) } \right. \nonumber \\
    &\phantom{\mathbb{P}} \left. \geq \sqrt{\mathcal{p}_i \left( 1-\mathcal{p}_i \right)}\right) \to 2c' \text{ as } k \to \infty.
\end{align}
Hence, for some $K > 0$, and for all $k > K$, for some constant $c > 0$ close to $c'$,
\begin{align}
	&\mathbb{P}\left( \abs{ \sum_{m \in w_{j_k}} \frac{q_\sigma(m-n_k)}{\sqrt{\sum_{m \in w_{j_k}} q_\sigma(m-n_k)^2}} \left(\mathbbm{1}_{X_m = i} - \mathcal{p}_i\right)  } \right. \nonumber \\
    &\phantom{\mathbb{P}} \left. \geq \sqrt{\mathcal{p}_i \left( 1-\mathcal{p}_i \right) }\right) \geq 2c. \label{eq:clt}
\end{align}
Rearranging the terms, and utilizing the symmetry of the Gaussian distribution, we see that for large enough $k$,
\begin{equation}
	\mathbb{P}\left( U_{ij_kn_k} \geq \E{U_{ij_kn_k}} + Q_{ij_kn_k} \right) \geq c.
\end{equation}
Replacing $k$ with $k'$ and with the same argument, we see that for some $K' > 0$, and for all $k > K'$,
\begin{equation}
	\mathbb{P}\left( U_{ij_kn_k} \leq \E{U_{ij_kn_k}} - Q_{ij_kn_k} \right) \geq c.
\end{equation}
Recall that $j_k$ and $n_k$ are minimizer of $\mathbb{P}\left( U_{ijn} > \E{U_{ijn}} + Q_{ijn}\right)$ (resp., $j_{k'}$ and $n_{k'}$ for $\mathbb{P}\left( U_{ijn} \leq \E{U_{ijn}} - Q_{ijn} \right)$); we conclude that for large enough $k$, for some $c > 0$, 
\begin{align}
	&\min_{j \in I_0} \min_{n \in W_j} \mathbb{P}\left( U_{ijn} > \E{U_{ijn}} + Q_{ijn}\right) \geq c \\
	\text{ and } &\min_{j \in I_0} \min_{n \in W_j} \mathbb{P}\left( U_{ijn} \leq \E{U_{ijn}} - Q_{ijn}\right) \geq c.
\end{align}
We hence conclude that for large enough $\sigma$,
\begin{align}
	&\phantom{\geq} \min_{j \in I_0} \min_{n \in W_j} \EE{\Perp}{\abs{\sum_{m \in w_j} q_\sigma(m-n) \left(\mathbbm{1}_{X_m = i} - \mathbbm{1}_{\hat{X}_m = i}\right)}} \nonumber \\
	&\geq c \sqrt{p_0 (1-p_0) (e^{1/\sigma}-1) \frac{e^{2/\sigma} \left( 1 - e^{-2k/\sigma} \right)}{\left( e^{1/\sigma}+1 \right)^3}}.
\end{align}

%Therefore, for some $\Sigma > 0$ and for all $\sigma > \Sigma$,
%\begin{align}
%	&\phantom{\geq} \E{\abs{\sum_{m  \in w_0} q_\sigma(m) \left(\mathbbm{1}_{X_m = i} - \mathcal{p}_i\right)}} \nonumber \\
%    &\geq \frac{C}{2} \sqrt{\mathcal{p}_i\left(1-\mathcal{p}_i\right) \cdot \sum_{m  \in w_0} q_\sigma(m)^2}
%\end{align}
%for some constant $C > 0$.

%Then, by Khintchine Inequality \cite[Section \S 10.3, Theorem 1]{chow2003probability}, 
%\begin{equation}
%	\E{\abs{\sum_{i  \in w_0} Q_i \left(X_i - \frac{1}{2}\right)}}  \geq A \left(\sum_{i  \in w_0} \frac{Q_i}{2} ^2\right)^{1/2} = A' \left(\sum_{i  \in w_0} Q_i ^2\right)^{1/2}
%\end{equation}
%for some constant $A > 1/2$ that does not change with $\sigma$.

For the second term in~\eqref{eq:windowexp}, by Pinsker's Inequality \cite{pinsker1964information}, for any $j \in I_0$,
\begin{align}
	&\phantom{=} \sum_{x_j^k,\hat{x}_j^k} \abs{p\left(x_j^k,\hat{x}_j^k\right) - p\left(x_j^k\right)p\left(\hat{x}_j^k\right)} \nonumber \\
    &= 2D_{\mathrm{TV}}\left(p\left(X_j^k,\hat{X}_j^k\right),p\left(X_j^k\right)p\left(\hat{X}_j^k\right)\right) \\
	&\leq \sqrt{2 D_{\mathrm{KL}}\left(p\left(X_j^k,\hat{X}_j^k\right)\Vert p\left(X_j^k\right)p\left(\hat{X}_j^k\right)\right)} \\
	&= \sqrt{2I(X_j^k;\hat{X}_j^k)} \\
	&\leq \sqrt{2/\sigma^{1+(\epsilon-\eta)/2}};
\end{align}
hence we conclude
\begin{align}
	&\phantom{\leq} \max_{j \in I_0} \max_{n \in W_j} \sum_{x_j^k,\hat{x}_j^k} \abs{p\left(x_j^k,\hat{x}_j^k\right) - p\left(x_j^k\right)p\left(\hat{x}_j^k\right)} \nonumber \\ 
	&\leq \sqrt{2/\sigma^{1+(\epsilon-\eta)/2}}.
\end{align}

For the third term in~\eqref{eq:windowexp}, for any $j \in I_0$ and for any $n \in W_j$,
\begin{align}
	&\phantom{\leq} \E{\abs{\sum_{m \notin w_j} q_\sigma(m-n) \left(\mathbbm{1}_{X_m = i} - \mathbbm{1}_{\hat{X}_m = i}\right)}} \nonumber \\
    &\leq \sum_{m \notin w_j} q_\sigma(m-n) \leq e^{-\sigma^{\eta/2}}\frac{2e^{1/\sigma}}{e^{1/\sigma}+1}.
\end{align}
hence
\begin{align}
	&\phantom{\leq} \max_{j \in I_0} \max_{n \in W_j} \E{\abs{\sum_{m \notin w_j} q_\sigma(m-n) \left(\mathbbm{1}_{X_m = i} - \mathbbm{1}_{\hat{X}_m = i}\right)}} \nonumber \\
    &\leq e^{-\sigma^{\eta/2}}\frac{2e^{1/\sigma}}{e^{1/\sigma}+1}. \label{eq:constantoutsidewindow}
\end{align}

Combining all of the above, we see that for large $\sigma$,
\begin{align}
	&\phantom{\geq} \min_{j \in I_0} \min_{n \in W_j} \E{W_2^2(Y_n,\hat{Y}_n)} \nonumber \\
	&\geq \frac{d_{\min}^2}{2} A \times  \\
	&\phantom{\geq} \left(c \sqrt{p_0 (1-p_0) (e^{1/\sigma}-1) \frac{e^{2/\sigma} \left( 1 - e^{-2k/\sigma} \right)}{\left( e^{1/\sigma}+1 \right)^3}} \right. \nonumber \\
	&\phantom{\geq} \ \left. - \sqrt{\frac{2}{\sigma^{1+(\epsilon-\eta)/2}}} - e^{-\sigma^{\eta/2}}\frac{2e^{1/\sigma}}{e^{1/\sigma}+1} \right), \nonumber 
\end{align}
and hence
\begin{align}
	&\phantom{\geq} \E{D_{\mathcal{d}_{\min}}} \\
	&\geq \frac{\floor{k\left( 1-\frac{2}{\sigma^{\eta/2}} \right)} |I_0|}{2N+1}\frac{d_{\min}^2}{2}\cdot A \times \\
	&\phantom{\geq} \left(c \sqrt{p_0 (1-p_0) (e^{1/\sigma}-1) \frac{e^{2/\sigma} \left( 1 - e^{-2k/\sigma} \right)}{\left( e^{1/\sigma}+1 \right)^3}} \right. \nonumber \\
	&\phantom{\geq} \ \left. - \sqrt{\frac{2}{\sigma^{1+(\epsilon-\eta)/2}}} - e^{-\sigma^{\eta/2}}\frac{2e^{1/\sigma}}{e^{1/\sigma}+1} \right).
\end{align}

Now, for an arbitrary random variable $M$,
\begin{align}
	&\phantom{=} (2N+1)R = I\left(X_{-N}^N;M\right) \\
    &= I\left(X_{j_1}^k,X_{j_2}^k,\ldots,X_{j_{\floor{(2N+1)/k}}}^k,X_r^\cdot;M\right), \\
\intertext{where $I = \{j_1,j_2,\ldots,j_{\floor{(2N+1)/k}}\}$, and $X_r^\cdot$ denotes all variables in the window of remainders. We can further write}
	&= \sum_{\ell = 1}^{\floor{(2N+1)/k}} I\left(X_{j_\ell}^k;M|X_{j_1}^k,\ldots,X_{j_{\ell-1}}^\cdot\right) \nonumber \\
	&\phantom{=} \ + I(X_r^\cdot;M|X_{j_1}^k,\ldots,X_{j_\ell}^k) \\
	&\geq \sum_{\ell = 1}^{\floor{(2N+1)/k}} I\left(X_{j_\ell}^k;M\right).
 \end{align}
Plugging in that the rate $R = 1/\sigma^{2+\epsilon}$, we have
\begin{equation}
	\frac{k}{\sigma^{2+\epsilon}} = kR \geq \frac{1}{(2N+1)/k} \sum_{\ell = 1}^{\floor{(2N+1)/k}} I\left(X_{j_\ell}^k;\hat{X}_{j_\ell}^k\right).
%    &= \frac{1}{(2N+1)/k} \times \nonumber \\
%    &\phantom{=} \left(\sum_{\ell = 1}^{\floor{(2N+1)/k}} D_{\mathrm{KL}}\left(p\left(X_{j_\ell}^k,\hat{X}_{j_\ell}^k\right)\Vert p\left(X_{j_\ell}^k\right)p\left(\hat{X}_{j_\ell}^k\right)\right)\right).
\end{equation}
Recall that $k = \sigma^{1+\eta}$ for some $\eta \leq \epsilon$; thus, 
\begin{align}
	&\phantom{=} 1/\sigma^{1+\epsilon - \eta} \geq \frac{1}{(2N+1)/k} \sum_{j \in I} I\left(X_j^k;\hat{X}_j^k\right) \\
	&= \frac{k}{2N+1}\floor{\frac{2N+1}{k}} \frac{1}{\floor{(2N+1)/k}} \sum_{j \in I} I\left(X_j^k;\hat{X}_j^k\right),
\end{align}
which implies that
\begin{equation}
	\frac{1}{\floor{(2N+1)/k}} \sum_{j \in I} I\left(X_j^k;\hat{X}_j^k\right) \leq \frac{1}{\frac{k}{2N+1}\floor{\frac{2N+1}{k}}} \frac{1}{\sigma^{1+\epsilon - \eta}}. \label{eq:avemi} 
\end{equation}
We see the left hand side of~\eqref{eq:avemi} is the average mutual information of all length $k$ windows, and the right hand side diminishes to $0$ as $N, \sigma \to \infty$. Let $P_\sigma = 1 - |I_0|/|I|$, i.e., the proportion of windows who violate the order bound in the definition of $I_0$. As $N, \sigma \to \infty$, we see that $P_\sigma \to 0$ as $1/\sigma^{(\epsilon-\eta)/2}$, otherwise~\eqref{eq:avemi} would fail. Equivalently, $|I_0| = |I|(1-P_\sigma) = \floor{\frac{2N+1}{k}}(1-P_\sigma)$.

We conclude that
\begin{align}
	&\phantom{\geq} \E{D_{\mathcal{d}_{\min}}} \nonumber \\
	&\geq \frac{1}{2N+1} \floor{k\left( 1-\frac{2}{\sigma^{\eta/2}} \right)} \floor{\frac{2N+1}{k}} (1-P_\sigma) \frac{d_{\min}^2}{2}\cdot A \times \nonumber \\
	&\phantom{\geq} \left(c \sqrt{p_0 (1-p_0) (e^{1/\sigma}-1) \frac{e^{2/\sigma} \left( 1 - e^{-2k/\sigma} \right)}{\left( e^{1/\sigma}+1 \right)^3}} \right. \nonumber \\
	&\phantom{\geq} \ \left. - \sqrt{\frac{2}{\sigma^{1+(\epsilon-\eta)/2}}} - e^{-\sigma^{\eta/2}}\frac{2e^{1/\sigma}}{e^{1/\sigma}+1} \right), \\
\intertext{which implies}
	&\phantom{\geq} \lim_{N \to \infty} \E{D_{\mathcal{d}_{\min}}} \nonumber \\
	&\geq \frac{1}{k} \floor{k\left( 1-\frac{2}{\sigma^{\eta/2}} \right)} (1-P_\sigma) \frac{d_{\min}^2}{2}\cdot A \times \nonumber \\
	&\phantom{\geq} \left(c \sqrt{p_0 (1-p_0) (e^{1/\sigma}-1) \frac{e^{2/\sigma} \left( 1 - e^{-2k/\sigma} \right)}{\left( e^{1/\sigma}+1 \right)^3}} \right. \nonumber \\
	&\phantom{\geq} \ \left. - \sqrt{\frac{2}{\sigma^{1+(\epsilon-\eta)/2}}} - e^{-\sigma^{\eta/2}}\frac{2e^{1/\sigma}}{e^{1/\sigma}+1} \right). \label{eq:lowerbound}
\end{align}
Using the sandwich argument~\ref{sec:sandwich}, we conclude that the same lower bound applies to $\E{D}$. Let $N \to \infty$, recall that $k = \floor{\sigma^{1+\eta}}$ and let $\sigma \to \infty$, the right hand side of~\eqref{eq:lowerbound} behaves like $1/\sigma^{(1-\eta)/2}$. We conclude that $(\alpha,\beta)$ is not asymptotically achievable when $\alpha < -2, \beta \leq -1/2$.
\end{proof}

\subsection{Sandwich Argument}
\label{sec:sandwich}

We provide a generic sandwich argument that deals with the distortion measure:

Define $d_{\max} = \max\{d_{i,j}: i \neq j\}$ to be the largest off-diagonal entry in $\mathcal{d}$, and $d_{\min} = \min\{d_{i,j}: i \neq j\}$ to be the smallest off-diagonal entry in $\mathcal{d}$. Define 
\begin{equation}
	\mathcal{d}_{\max} = \begin{bmatrix} 0 & & d_{\max} \\ & \ddots & \\ d_{\max} & & 0 \end{bmatrix}
\end{equation}
and
\begin{equation}
	\mathcal{d}_{\min} = \begin{bmatrix} 0 & & d_{\min} \\ & \ddots & \\ d_{\min} & & 0 \end{bmatrix},
\end{equation}
i.e. $\mathcal{d}_{\max}$ (resp., $\mathcal{d}_{\min}$) is the distance matrix where all off-diagonal entries are replaced by $d_{\max}$ (resp., $d_{\min}$).
We can write
\begin{equation}
W_{2,\mathcal{d}}^2(Y_n,\hat{Y}_n) = \inf_{\pi \in \Pi(Y_n,\hat{Y}_n)} \sum_{i \neq j} d_{i,j}^2 \pi(i,j), \label{eq:jointdist}
\end{equation}
where $\Pi(Y_n,\hat{Y}_n)$ is the collection of all joint distributions such that the marginal distributions are $Y_n$ and $\hat{Y}_n$, respectively. Let $\pi'$ be the joint distribution that achieves $W_{2,\mathcal{d}}^2(Y_n,\hat{Y}_n)$, i.e., $\pi' = \arg\min_{\pi \in \Pi\left(Y_n,\hat{Y}_n\right)} \sum_{i,j} d_{i,j}^2 \pi(i,j)$; then, (\ref{eq:jointdist}) can be bounded by
\begin{align}
\inf_{\pi \in \Pi\left(Y_n,\hat{Y}_n\right)} \sum_{i \neq j} d_{\min}^2 \pi(i,j) &\leq \sum_{i \neq j} d_{\min}^2 \pi'(i,j)  \nonumber \\
\leq \sum_{i \neq j} d_{i,j}^2 \pi'(i,j) &\leq \sum_{i \neq j} d_{i,j}^2 \pi''(i,j) \\
\leq \sum_{i \neq j} d_{\max}^2 \pi''(i,j) &= \inf_{\pi \in \Pi\left(Y_n,\hat{Y}_n\right)} \sum_{i \neq j} d_{\max}^2 \pi(i,j), \nonumber
\end{align}
where $\pi'' = \arg\min_{\pi \in \Pi\left(Y_n,\hat{Y}_n\right)} \sum_{i,j} d_{\max}^2 \pi(i,j)$; i.e., $\pi''$ is the optimal coupling under $\mathcal{d}_{\max}$. We conclude that 
\begin{equation}
W_{2,\mathcal{d}_{\min}}^2(Y_0,\hat{Y}_0) \leq W_{2,\mathcal{d}}^2(Y_0,\hat{Y}_0) \leq W_{2,\mathcal{d}_{\max}}^2(Y_0,\hat{Y}_0),
\end{equation}
where we use an extra subscript to specify the corresponding underlying distortion metric to the Wasserstein distance. Notice that the constants are different ($d_{\min}$ and $d_{\max}$, respectively); in the achievability and converse proofs, the constants needs to be changed when applying the sandwich argument; however this does not change the order laws. Summing them up, we see
\begin{equation}
D_{\mathcal{d}_{\min}} \leq D \leq D_{\mathcal{d}_{\max}},
\end{equation}
where we use an extra subscript for Wasserstein distortion under the different metric.

%\section{Conclusion}
%
%placeholder

%%%%%%
%% Appendix:
%% If needed a single appendix is created by
%%
%\appendix
%%
%% If several appendices are needed, then the command
%%
% \appendices
%%
%% in combination with further \section commands can be used.
%%%%%%

\section*{Acknowledgment}

The authors wish to thank Johannes Ball\'{e} and Lucas Theis
for helpful discussions.
This research was supported by the US National Science Foundation under grant CCF-2306278 and a gift from Google.

%%%%%%
%% To balance the columns at the last page of the paper use this
%% command:
%%
%\enlargethispage{-1.2cm} 
%%
%% If the balancing should occur in the middle of the references, use
%% the following trigger:
%%
\IEEEtriggeratref{4}
%%
%% which triggers a \newpage (i.e., new column) just before the given
%% reference number. Note that you need to adapt this if you modify
%% the paper.  The "triggered" command can be changed if desired:
%%
%\IEEEtriggercmd{\enlargethispage{-20cm}}
%%
%%%%%%

\bibliographystyle{IEEEtran}
\bibliography{bib}

\newpage
\appendices

\section{Proofs for Section~\ref{sec:extreme}}
\label{app:proofs}

\begin{proof}[Proof of Theorem~\ref{thm:fidelity}]
Fix $K$ and $\epsilon$ as in \ref{TSG:tail}. Consider the coupling between $y_{0,\sigma}$ and $\hat{y}_{0,\sigma}$ suggested by the ordering of the sequences:
\begin{align}
D_{0,\sigma} & = \inf_{Z \sim y_{n,\sigma},\hat{Z} \sim \hat{y}_{n,\sigma}} E[d^p(Z,\hat{Z})] \nonumber \\
 & \le \sum_{k = -\infty}^\infty q_{\sigma}(k) d^p(z_k,\hat{z}_k).
\end{align}
We have
\begin{align}
\limsup_{\sigma \rightarrow 0} D_{0,\sigma}
   & \le \lim_{\sigma \rightarrow 0} q_\sigma(0) d^p(z_0,\hat{z}_0) \\
         \nonumber
   & \phantom{ = \lim } + \lim_{\sigma \rightarrow 0} \sum_{k: 0 < |k| \le K} q_\sigma(k) d^p(z_k,\hat{z}_k) \\
         \nonumber
   & \phantom{ = \lim } + \lim_{\sigma \rightarrow 0} \sum_{k: |k| > K} q_\sigma(k) d^p(z_k,\hat{z}_k) \\
      & = d^p(z_0,\hat{z}_0),
      \label{eq:fidelity:dom}
\end{align}
where (\ref{eq:fidelity:dom}) follows from \ref{TSG:delta} and \ref{TSG:cont} (for the first two limits) and from \ref{TSG:delta}-\ref{TSG:tail} and dominated convergence (for the third limit). For the reverse direction, fix $\sigma > 0$ and let $q_\sigma(\cdot,\cdot)$ denote any PMF over $\mathbb{Z}^2$, both of whose marginals are $q_\sigma(\cdot)$. Then we have
\begin{align}
&\phantom{\geq}\sum_{k_1 = - \infty}^\infty \sum_{k_2 = -\infty}^\infty q_\sigma(k_1,k_2) d^p(z_{k_1},\hat{z}_{k_2}) \nonumber \\
 & \ge q_{\sigma}(0,0) d^p(z_0,\hat{z}_0) \\
 & \ge (2q_{\sigma}(0)-1) d^p(z_0,\hat{z}_0),
\end{align}
from which the result follows by \ref{TSG:delta} and \ref{TSG:cont}.
\end{proof}

To prove Theorem~\ref{thm:realism}, we need a lemma first.

\begin{lem}[Equivalence of Ces\`{a}ro Sums]\label{infinitylemma}
Suppose $q$ satisfies \ref{TSG:symm}, \ref{TSG:mono}, and \ref{TSG:unif}. For any two-sided $\mathbb{R}$-valued sequence $\seq{a}$, if  
\begin{equation}
	\lim_{m \to \infty} \frac{1}{2m+1} \sum_{k = -m}^m a_k = \alpha \in \mathbb{R},
\end{equation}
%\begin{equation}
%	\sum_{\ell = -\infty}^\infty a_\ell < \infty;
%\end{equation}
 and for all $\sigma > 0$,
 \begin{equation}
 \label{eq:cesaro:summable}
 	\sum_{k = -\infty}^\infty  q_\sigma(k) |a_k| < \infty,
 \end{equation}
then we likewise have
\begin{equation}
	\lim_{\sigma \to \infty} \sum_{k = -\infty}^\infty q_\sigma(k) a_k = \alpha.
\end{equation}
\label{lem:cesaro}
\end{lem}

\begin{proof}[Proof of Lemma~\ref{lem:cesaro}]
We can write
\begin{align}
   \label{eq:cesaro:dom}
    &\phantom{=}\sum_{\ell = -\infty}^\infty q_\sigma(\ell) a_\ell \nonumber \\
	&= \lim_{k \to \infty} \sum_{\ell = -(k-1)}^{k-1} \left(q_\sigma(\ell) - q_\sigma(k)\right) a_\ell \\
	&= \lim_{k \to \infty} \sum_{\ell = -(k-1)}^{k-1} \sum_{m = \abs{\ell}}^{k-1} \left(q_\sigma(m) - q_\sigma(m+1)\right) a_\ell \\
	&= \lim_{k \to \infty} \sum_{m = 0}^{k-1} \sum_{\ell = -m}^m \left(q_\sigma(m) - q_\sigma(m+1)\right) a_\ell \\
	\label{eq:cesaro:ID}
	&= \sum_{m = 0}^{\infty} \sum_{\ell = -m}^m \left(q_\sigma(m) - q_\sigma(m+1)\right) a_\ell,
\end{align}
where (\ref{eq:cesaro:dom}) holds by (\ref{eq:cesaro:summable}), \ref{TSG:symm}, \ref{TSG:mono}, and dominated convergence. For $m \ge 0$, define the sequences
\begin{align}
	b_m & = \frac{1}{2m+1}\sum_{\ell = -m}^m a_\ell
\intertext{and}
	r_\sigma(m) & = \left(q_\sigma(m) - q_\sigma(m+1)\right)(2m+1).
\end{align}
By \ref{TSG:mono}, $r_\sigma(m) \ge 0$. By (\ref{eq:cesaro:ID}),
\begin{equation}
	\sum_{\ell = -\infty}^\infty q_\sigma(\ell) a_\ell = \sum_{m = 0}^\infty r_\sigma(m) b_m.
\end{equation}
Now the choice $a_\ell = 1$ satisfies (\ref{eq:cesaro:summable}) and in this case the previous equation reads $\sum_{m = 0}^\infty r_\sigma(m) = 1$. Fix $\epsilon > 0$ and $M$ such that for all $m > M$, $\abs{b_m - \alpha} < \epsilon$. We can write
\begin{align}
	&\phantom{\leq} \abs{\sum_{\ell = -\infty}^\infty q_\sigma(\ell) a_\ell - \alpha} \nonumber \\
	&\leq \abs{\sum_{m = 0}^M r_\sigma(m) (b_m - \alpha)} \nonumber \\
	&\phantom{\leq} + \abs{\sum_{m = M+1}^\infty r_\sigma(m) (b_m - \alpha)}\\
	&\leq \left(\sum_{m = 0}^M r_\sigma(m)\right)\left(\max_{m = 0,1,\ldots,M} b_m + \abs{\alpha}\right) + \epsilon.
\end{align}
Taking $\sigma \to \infty$ on both sides, the conclusion follows by \ref{TSG:unif}.
\end{proof}

We now prove Theorem~\ref{thm:realism}.

\begin{proof}[Proof of Theorem~\ref{thm:realism}]
With a slight abuse of notation, let $F_\sigma$ denote the CDF of the distribution
\begin{equation}
    \sum_{k = -\infty}^\infty q_\sigma(k) \delta_{z_k},
\end{equation}
and define $\hat{F}_\sigma$ analogously. Then $D_{0,\sigma} = W_p^p(F_\sigma,\hat{F}_\sigma)$. By the triangle inequality for Wasserstein distance~\cite[p.~94]{villani2009optimal} (which requires $d$ to be a metric),
\begin{equation}
W_p(F_\sigma,\hat{F}_\sigma) \le W_p(F_\sigma,F) + W_p(F,\hat{F}) + W_p(\hat{F},\hat{F}_\sigma).
\end{equation}
By Lemma~\ref{infinitylemma} and (\ref{eq:realism:CDF}), $F_\sigma \stackrel{w}{\rightarrow} F$. By Lemma~\ref{infinitylemma}, (\ref{eq:realism:moment}), and (\ref{eq:realism:summable}), we have
\begin{equation}
 %   \lim_{\sigma \rightarrow \infty} \sum_{n = -\infty}^\infty q_\sigma(n) d^p(z_n,0)
 \lim_{\sigma \rightarrow \infty} 
         \int d^p(z,0) dF_\sigma(z)
         = \int d^p(z,0) dF(z).
\end{equation}
These two conditions imply that $W_p(F_\sigma,F) \rightarrow 0$ as $\sigma \rightarrow \infty$ \cite[Thm.~6.9]{villani2009optimal}. Similarly we have $W_p(\hat{F},\hat{F}_\sigma) \rightarrow 0$, yielding 
\begin{equation}
\limsup_{\sigma \rightarrow \infty} D_{0,\sigma} \le W_p^p(F,\hat{F}).
\end{equation}
Applying the triangle inequality in the reverse direction gives
\begin{equation}
    W_p(F,\hat{F}) \le W_p(F,F_\sigma) + W_p(F_\sigma,\hat{F}_\sigma) + W_p(\hat{F}_\sigma,\hat{F}).
\end{equation}
Taking limits yields 
\begin{equation}
\liminf_{\sigma \rightarrow \infty} D_{0,\sigma} \ge W_p^p(F,\hat{F})
\end{equation}
and the theorem.
\end{proof}

\begin{proof}[Proof of Corollary~\ref{cor:ergodic}]
Among the hypotheses of Theorem~\ref{thm:realism}, (\ref{eq:realism:CDF}) and (\ref{eq:realism:moment}) hold a.s.\ by the ergodic theorem (e.g.,~\cite[Thm 6.2.1]{durrett1996pte}) and (\ref{eq:finitepthmoment}), and (\ref{eq:realism:summable}) holds a.s.\ because
\begin{align}
   \E{\sum_{k = -\infty}^\infty q_\sigma(k) d^p(Z_k,0)}
     &= \sum_{k = -\infty}^\infty q_\sigma(k) \E{d^p(Z_k,0)} \nonumber \\
	& = \E{d^p(Z_0,0)} < \infty,
\end{align}
by monotone convergence and (\ref{eq:finitepthmoment}), and similarly for $\hat{Z}$.
%multivariate Glivenko-Cantelli
%theorem~\cite{wright1981empirical} 
%and (\ref{eq:averagemoment}) holds a.s.\ by (\ref{eq:finitepthmoment}) and similarly for $\hat{Z}$.
\end{proof}

\end{document}